\documentclass[a4paper,11pt,reqno]{amsart}

\usepackage{amssymb,amsmath}
\usepackage{mathrsfs,bbm}
\usepackage{color}

\newtheorem{thm}{Theorem}[section]
\newtheorem{prp}[thm]{Proposition}
\newtheorem{lem}[thm]{Lemma}
\newtheorem{dfn}[thm]{Definition}

\newtheorem{cor}[thm]{Corollary}
\newtheorem{example}[thm]{Example}

\newtheorem{remark}[thm]{Remark}

\newenvironment{rmk}{\begin{remark} \rm }{\hfill $\Box$ \end{remark}}
\newenvironment{prf}{\noindent {\it Proof:} \ }{\hfill $\Box$}

\newcommand\od{\mathrm{d}}
\newcommand\p{\partial}

\newcommand{\nn}{\nonumber}

\newcommand{\ld}{\lambda} 

\newcommand{\gm}{\gamma}
\newcommand{\sg}{\sigma}
 
\newcommand{\Gm}{\Gamma}

\newcommand{\dt}{\delta}
 
\newcommand{\ka}{\kappa}

\newcommand{\ve}{\varepsilon}

\newcommand{\res}{\mathrm{Res}}

\newcommand\Z{\mathbb{Z}}

 \newcommand\cB{\mathcal{B}}

\newcommand\cE{\mathcal{E}}

\newcommand{\set}[1]{\left\{#1\right\}}

\newcommand\ra{\right\rangle}
\newcommand\la{\left\langle}

\newcommand{\bm}[1]{\mathbf{#1}}

\newcommand{\bt}{\bm{t}}  

\allowdisplaybreaks \numberwithin{equation}{section}

\begin{document}

\title[Additional symmetries of KP-mKP]{Additional symmetries of the KP-mKP hierarchy and Virasoro constraints to the Burgers-KdV hierarchy }
\author{
Zongyao Feng$^\dag$ \quad Lumin Geng$^\star$ 
\quad Chao-Zhong Wu$^\dag$ 
}
\dedicatory {
	{\small $\dag$ School of Mathematics, Sun Yat-Sen University, Guangzhou 510275, P. R. China}
\\
{\small $\star$ School of Mathematics and Systems Science, Guangdong Polytechnic Normal University, Guangzhou 510665, P. R. China; Email address:  genglm@gpnu.edu.cn
}
 }


\begin{abstract}
A KP-mKP hierarchy was introduced recently via pseudo-differential operators containing two derivations. In this paper, for the KP-mKP hierarchy we derive a class of (differential) Fay identities and construct a series of additional symmetries. Moreover, the additional symmetries are represented as certain linear actions on the tau functions of the hierarchy, with the help of the 	Adler-Shiota-van Moerbeke formula. As an application, we reprove the Virasoro constraints to the tau functions of the Burgers-KdV hierarchy, and such results are generalized to its higher order extensions regarded as reductions of the KP-mKP hierarchy.
\end{abstract}
\maketitle

\section{Introduction}

The Kadomtsev-Petviashvili (KP) hierarchy \cite{Sato} and the (first) modified KP (mKP) hierarchy \cite{JM1983} (see also \cite{Kac-vdL-mKP,Kiso1990,Kup1985,OevelG1993}) are among the most fundamental models in the theory of integrable systems, and they have been investigated from various points of view. For instance, for both hierarchies there exists a class of non-isospectral symmetries named as additional symmetries, which commute with all flows of the hierarchy but do not commute
between themselves (generating a $w_{\infty}$-algebra instead) \cite{AvM,Cheng2018,Dickey1995}.
Generally speaking, additional symmetries can be constructed via the so-called Orlov-Schulman operators \cite{OS1986}, or via infinitesimal actions on the tau functions by certain vertex operators \cite{DKJM-KPBKP}. These two approaches were shown equivalent for the KP hierarchy by Adler, Shiota and van Moerbeke \cite{ASvM94} and also by Dickey \cite{Dickey1995}, based on the celebrated Adler-Shiota-van Moerbeke (ASvM) formula. As is known, additional symmetries of integrable hierarchies are closely related to research topics such as string equation, Virasoro symmetries and $W$-constraints in matrix models and quantum field theory (see, for example, \cite{AvM,AvM99,ANP1997,BW,Dickey,Wu2013} and references therein).

Recently, a KP-mKP hierarchy was introduced \cite{GHW2024} via pseudo-differential operators with two derivations. This hierarchy can be recast to the form of Hirota bilinear equations of two tau functions $\tau_1$ and $\tau_2$, and is equivalent to the so-called extended KP hierarchy \cite{WZ2016} under a certain generic assumption, which is regarded as the dispersive Whitham hierarchy on the Riemann sphere with the infinity point and a movable point marked \cite{KIM1994,SB2008}. Moreover, under certain constraints the KP-mKP hierarchy can be reduced to the KP and the mKP hierarchies, the two-component BKP hierarchy, as well as the Burgers-KdV hierarchy \cite{Bur2015,Bur2016,GHW2025}.
In consideration of such reduction properties, it is natural to ask whether there is a uniform way to understand the additional symmetries of the KP and of the mKP hierarchy. A trial was made in \cite{LW2021}, where the additional symmetries of the extended KP hierarchy have been constructed via Orlov-Schulman operators containing a single derivation, however, the ASvM formula and the representation with tau functions for the additional symmetries are still missing.

In the present paper, we are to resolve the above question in the framework of the KP-mKP hierarchy. For preparation, we will firstly derive a class of identities of Fay type for the hierarchy, and clarify some properties of its Baker-Akhiezer functions and tau functions. For instance, we will show that both tau functions $\tau_1$ and $\tau_2$ of the KP-mKP hierarchy indeed obey the Hirota bilinear equations of the KP hierarchy with half time variables of the former. Secondly,
we will construct a series of additional symmetries of the KP-mKP hierarchy with the help of certain Orlov-Schulman operators, which will be seen to give a realization of the $w_{\infty}\times w_{\infty}$-algebra; see Theorem~\ref{thm-addc} below for details. With the help of (differential) Fay identities of the KP-mKP hierarchy, we will derive the ASvM formula (see Theorem~\ref{thm-asvm} below) of the KP-mKP hierarchy, which identify infinitesimal transformations of the Baker-Akhiezer functions with generating series of additional symmetries. Accordingly, we will achieve our main result (see Theorem~\ref{thm-tau} below), that is, a representation of the additional symmetries of the KP-mKP hierarchy into some linear actions on the tau functions of the form
\begin{align}
	\frac{\p \tau_\mu }{\p s_{m,m+l}^{(\nu)}}=&\left(\frac{W_{m+1,l}^{(\nu)}}{m+1}+(1-\delta_{\nu,\mu})W_{m,l}^{(\nu)} +\delta_{l0}c_m^{(\nu)}\right)\tau_\mu, \label{add-tau2}
\end{align}
where $\nu,\mu\in\{1,2\}$, $m\in\Z_{\geq0},\,l\in\Z$, $c_m^{(\nu)}$ are constants, and the operators $W_{m,l}^{(\nu)}$ defined in \eqref{W} below are independent of $\tau_\mu$. Such results, under appropriate reductions, coincide with those for the KP and the mKP hierarchies obtained in \cite{AvM,Cheng2018,Dickey1995}.

As an application of the formulae \eqref{add-tau2}, we will derive a series of Virasoro symmetries of the Burgers-KdV hierarchy \cite{Bur2015}, which can be viewed as a certain reduction of the KP-mKP hierarchy \cite{Ale,GHW2025}. Moreover, we will give another proof of the Virasoro constraints to the tau functions of the Burgers-KdV hierarchy starting from the string equations (see Theorem~\ref{thm-2KPvir} below). Note that such kind of ``open Virasoro equations'' (cf. equations (1.11) in \cite{Bur2015}) were conjectured by Pandharipande, Solomon and Tessler \cite{PST} and proved by Buryak (see Theorem~1.2 in \cite{Bur2015}), with the help of the recursion relations of the Burgers-KdV hierarchy. What is more, in the framework of the KP-mKP hierarchy we will also prove the Virasoro constraints to the two tau functions of a higher order extension of the Burgers-KdV hierarchy, the so-called extended $r$-reduced KP hierarchy proposed by Bertola and Yang \cite{BY2015}.

This paper is arranged as follows. In Section 2, we will recall some notations of pseudo-differential operators containing two derivations and the definition of the KP-mKP hierarchy. In Section 3, we will derive the (differential) Fay identities of the KP-mKP hierarchy, which will help us to show that each of the tau functions $\tau_1$ and $\tau_2$ indeed solves the KP hierarchy. In Sections 4 and 5, for the KP-mKP hierarchy we will construct its additional symmetries, derive the ASvM formula, and prove the equalities \eqref{add-tau2}. In section 6, we will derive the Virasoro constraints to the Burgers-KdV hierarchy and its higher order extensions. The final section is devoted to some remarks.

\section{Pseudo-differential operators and the KP-mKP hierarchy}
In this section, we are to recall some properties of pseudo-differential operators with two derivations and the KP-mKP hierarchy, following the notations in \cite{GHW2024}.

\subsection{Pseudo-differential operators with two derivations}

Let $\cB$ be a commutative associative algebra of smooth complex functions of two variables $x_1$ and $x_2$, which is acted commutatively by the following two derivations:
\[\p_\nu=\frac{\od}{\od x_\nu},\quad \nu\in\{1,2\}.
\]
Consider an algebra of pseudo-differential operator with two derivations $\p_1$ and $\p_2$ over $\cB$ as follows:
\begin{align*}\label{}
	&\cE=\left\{ \sum_{i\le m}\sum_{j\le n}f_{i j} \p_1^i \p_2^j\mid f_{i j}\in\cB; \, m,n\in\Z \right\}.
\end{align*}
In this algebra the product is defined by
\begin{equation*}\label{}
	f \p_1^i\p_2^j\cdot g \p_1^p \p_2^q=\sum_{r,s\geq0}\binom{i}{r}\binom{j}{s} f\, \p_1^r\p_2^s(g)\cdot
	\p_1^{i+p-r}\p_2^{j+q-s},
\end{equation*}
with $f,\,g\in\cB$, and
\[
\binom{i}{r}=\frac{i(i-1)\dots(i-r+1)}{r!}.
\]
The Lie bracket between two pseudo-differential operators $A,B\in\cE$ is defined by the commutator
\[
[A, B]=A B-B A.
\]

Given an operator $A=\sum_{i\le m}\sum_{j\le n}  {f}_{i j} \p_1^i\p_2^j\in\cE$, its residues and adjoint operator mean
\begin{equation*}\label{resA}
	\res_{\p_1}A=\sum_{j\le n}  {f}_{-1, j} \p_2^j, \quad \res_{\p_2}A=\sum_{i\le m}  {f}_{i, -1} \p_1^i,
\quad
	A^*=\sum_{i\le m}\sum_{j\le n} (-1)^{i+j}  \p_1^i\p_2^j f_{i j}.
\end{equation*}
The action on $A$ by a differential operator
$D=\sum_{r,s\ge0}  {g}_{r s} \p_1^r\p_2^s\in\cE$ is written as
\begin{equation*}\label{exDA}
	D\la{A}\ra=\sum_{i,j} \left(\sum_{r,s\ge0}  {g}_{r s} \p_1^r \p_2^s (f_{i j})\right) \p_1^i\p_2^j,
\end{equation*}
i.e.   $D$ acts on each coefficient of $A$.

In particular, whenever $A$ involves only powers of the derivation $\p_1$ or of $\p_2$, the above notions agree with those for pseudo-differential operators containing a single derivation. In this case, for $A=\sum_{i}f_i\p_\nu^i\in\mathcal{E}$ with $f_i\in\cB$, one has truncations
\begin{equation*}\label{cutop}
	A_{\ge m}=\sum_{i\ge m} f_i \p_\nu^i, \quad A_{< m}=\sum_{i<m} f_i \p_\nu^i,\quad m\in\Z.
\end{equation*}

For $\nu\in\{1,2\}$ and $i\in\Z$, we assign the actions of the pseudo-differential operator on certain exponential functions with nonzero complex parameter $z$ as
\begin{align*}
	\p_\nu^i(e^{x_{\nu}z})=z^i e^{x_{\nu}z}.
\end{align*}
To shorten notations, we usually write $\p_\nu^ie^{x_{\nu}z}=\p_\nu^i(e^{x_{\nu}z})$.
The following result is useful below.
\begin{lem}[see, for example, \cite{DKJM-KPBKP}] \label{thm-exres}
	For any pseudo-differential operators $F, G\in\mathcal{E}$ that contain only powers in $\p_\nu$ with $\nu\in\{1,2\}$, the following equality holds true:
	\begin{equation*}
		-\res_{z=\infty}\left(F e^{z x_\nu }\cdot G^* e^{-z x_\nu}\right)=\res_{\p_\nu} (F G).
	\end{equation*}
\end{lem}

\subsection{The KP-mKP hierarchy}
Letting $a_{\nu,i},\beta\in\cB$, we consider two pseudo-differential operators of the form:
\begin{align}
	\Phi_1=1+\sum_{i\ge1}a_{1,i} \p_{1}^{-i},\quad
	\Phi_2=e^{\beta }\Big(1+\sum_{i\ge1}a_{2,i}\p_{2}^{-i}\Big). \label{re-Phi120}
\end{align}
Clearly these two operators are invertible. Let
\begin{align}
	P_1 =\Phi_1\p_1\Phi_1^{-1}=\p_1+\sum_{i\ge1}u_{1,i}\p_{1}^{-i},\quad
	P_2=\Phi_2\p_2\Phi_2^{-1}=\p_2-\p_2(\beta)+\sum_{i\ge1}u_{2,i}\p_{2}^{-i}. \label{exP2}
\end{align}
One sees that each coefficient $u_{\nu,i}$ is a differential polynomial in $a_{\nu,i},\beta$.

The KP-mKP hierarchy is defined by
the following evolutionary equations:
\begin{align}
	\frac{\p \Phi_\mu}{\p t_{\nu, k}}&=B_{\nu,k}^\mu\Phi_\mu,\label{Phinutmu}\\
	e^\beta\p_1 e^{-\beta}\Phi_1\p_2\Phi_1^{-1}&=\p_2\Phi_2\p_1\Phi_{2}^{-1}, \label{Phinumu}
\end{align}
where $k\in\Z_{\geq1},\,\nu,\mu\in\{1,2\}$, and
\begin{equation*}
	B_{\nu,k}^\mu=\left\{ \begin{array}{cl}
		-(P_\nu^k)_{<\nu-1}, & \mu=\nu; \\
		(P_\nu^k)_{\geq\nu-1}\la\Phi_\mu\ra \Phi_\mu^{-1}, & \mu\ne\nu.
	\end{array}\right.
\end{equation*}
This hierarchy is reduced to the KP hierarchy and the mKP hierarchy under the constraint $\Phi_2=e^\beta$ and $\Phi_1=1$ respectively.

Note that
${\p}/{\p t_{\nu,1}}=\p_{\nu}$,
we will just take
\[
t_{\nu,1}=x_\nu, \quad \nu\in\set{1,2}.
\]
For $\nu\in\{1,2\}$, denote $\bt_\nu=(t_{\nu,1},t_{\nu,2},t_{\nu,3},\dots)$ and
\[ 
\xi(\bt_\nu; z)=\sum_{k\in\Z_{\geq1}}
t_{\nu,k} z^k.
\] 
For the KP-mKP hierarchy, its Baker-Akhiezer functions take the form:
\begin{align}
w_\nu(\bt_1, \bt_2; z)=\Phi_\nu e^{\xi(\bt_\nu;z)}, \quad \nu\in\{1,2\}, \label{wavef}	
\end{align}
and the adjoint Baker-Akhiezer functions read:
\begin{equation}\label{adwavef}
	w_1^\dag(\bt_1,\bt_2;z) =\left(\p_1\Phi_1^{-1}e^{\beta}\p_1^{-1}e^{-\beta}\right)^*e^{-\xi(\bt_1;z)}, \quad
	w_2^\dag(\bt_1,\bt_2;z) =\left(\p_2\Phi_2^{-1}\p_2^{-1}\right)^*e^{-\xi(\bt_2;z)}.
\end{equation}
\begin{thm}[\cite{GHW2024}]\label{thm-ble2}
	(I) The KP-mKP hierarchy \eqref{Phinutmu}--\eqref{Phinumu} is equivalent to the following bilinear equation
	\begin{equation}\label{exble3}
		\res_{z=\infty}\left( z^{-1} w_1(\bt_1, \bt_2; z)w_1^{\dag}(\bt_1', \bt_2'; z) \right)=\res_{z=\infty} \left(z^{-1}w_2(\bt_1, \bt_2; z)w_2^{\dag}(\bt_1', \bt_2'; z)\right)
	\end{equation}
	with arbitrary time variables $(\bt_1, \bt_2)$ and $(\bt_1', \bt_2')$.\\
	(II) There exist two tau functions $\tau_1(\bt_1,\bt_2)$ and $\tau_2(\bt_1,\bt_2)=e^{\beta(\bt_1,\bt_2)}\tau_1(\bt_1,\bt_2)$ such that
	\begin{align}
		w_1(\bt_1,\bt_2;z)=&\frac{\tau_1(\bt_1-[z^{-1}],\bt_2)}{\tau_1(\bt_1,\bt_2)} e^{\xi(\bt_1;z)},\quad w_1^{\dag}(\bt_1,\bt_2;z)=\frac{\tau_2(\bt_1+[z^{-1}],\bt_2)}{\tau_2(\bt_1,\bt_2)} e^{-\xi(\bt_1;z)},\label{w1tau} \\ w_2(\bt_1,\bt_2;z)=&\frac{\tau_2(\bt_1,\bt_2-[z^{-1}])}{\tau_1(\bt_1,\bt_2)} e^{\xi(\bt_2;z)},\quad w_2^{\dag}(\bt_1,\bt_2;z)=\frac{\tau_1(\bt_1,\bt_2+[z^{-1}])}{\tau_2(\bt_1,\bt_2)}e^{-\xi(\bt_2;z)}, \label{w2tau}
	\end{align}
	where $[z^{-1}]=\left(\frac{1}{z}, \frac{1}{2z^2}, \frac{1}{3z^3}, \dots\right)$.
\end{thm}

\section{Identities of Fay type and their applications}
As a preparation for Section~\ref{sec-addtau} below, let us derive a class of (differential) Fay identities of the KP-mKP hierarchy. Based on them, we will also clarify some properties of the Baker-Akhiezer functions and tau functions of the  hierarchy.

\begin{thm}
	The tau functions $\tau_1(\bt_1,\bt_2)$ and $\tau_2(\bt_1,\bt_2)$ of the KP-mKP hierarchy satisfy the following Fay identity
	\begin{align}
	&	\sum_{\mathrm{c.p.}\{s_1^1,  s_2^1, s_3^1\}}\frac{s_1^1(s_1^1-s_0^1)}{(s_1^1-s_2^1)(s_1^1-s_3^1)}\tau_1(\bt_1+[s_2^1]+[s_3^1],\bt_2+[s_1^2]+[s_2^2]+[s_3^2]) \cdot \nn \\
& \qquad \qquad \cdot \tau_2(\bt_1+[s_0^1]+[s_1^1],\bt_2+[s_0^2])\nn\\
	=&	\sum_{\mathrm{c.p.}\{s_1^2,  s_2^2, s_3^2\}}\frac{s_1^2(s_1^2-s_0^2)}{(s_1^2-s_2^2)(s_1^2-s_3^2)}\tau_2(\bt_1+[s_1^1]+[s_2^1]+[s_3^1],\bt_2+[s_2^2]+[s_3^2]) \cdot \nn \\
& \qquad \qquad  \cdot\tau_1(\bt_1+[s_0^1],\bt_2+[s_0^2]+[s_1^2]).\label{kpmkp-fay}
	\end{align}
Here, for $\nu\in\{1,2\}$, $s_0^\nu,\,s_1^\nu,\,s_2^\nu,\,s_3^\nu$ are pairwise distinct nonzero complex parameters, $[s_i^\nu]=\left(\frac{s_i^\nu}{1},\frac{(s_i^\nu)^2}{2},\frac{(s_i^\nu)^3}{3},\cdots\right)$ (recall notations in Theorem~\ref{thm-ble2}), and ``$\mathrm{c.p.}$" stands for ``cyclic permutation".
\end{thm}
\begin{prf}
According to Theorem \ref{thm-ble2}, the tau functions of the KP-mKP hierarchy satisfy
\begin{align}
	&-\res_{z=\infty} z^{-1}\tau_1(\bt_1-[z],\bt_2)\tau_2(\bt_1'+[z],\bt_2')e^{\xi(\bt_1-\bt_1';z^{-1})}\od z\nn\\
	=&-\res_{z=\infty} z^{-1}\tau_2(\bt_1,\bt_2-[z])\tau_1(\bt_1',\bt_2'+[z])e^{\xi(\bt_2-\bt_2';z^{-1})}\od z.\label{ble-tau}
\end{align}
Let us do the following replacements of time variables:
\begin{align*}
\bt_\nu\mapsto \bt_\nu+[s_1^\nu]+[s_2^\nu]+[s_3^\nu],\quad \bt_\nu'\mapsto \bt_\nu+[s_0^\nu],
\end{align*}
then the left hand side of equation \eqref{ble-tau} reads
\begin{align*}
	\text{l.h.s.}=&-\res_{z=\infty} \frac{1-s_0^1/z}{z(1-s_1^1/z)(1-s_2^1/z)(1-s_3^1/z)}\cdot\nn\\
	&\qquad   \cdot \tau_1(\bt_1+[s_1^1]+[s_2^1]+[s_3^1]-[z],\bt_2+[s_1^2]+[s_2^2]+[s_3^2]) \tau_2(\bt_1+[s_0^1]+[z],\bt_2+[s_0^2])\nn\\
=&\sum_{i=1}^3\res_{z=s_i^1} \frac{z(z-s_0^1)}{(z-s_1^1)(z-s_2^1)(z-s_3^1)}\cdot\nn\\
	&\qquad   \cdot \tau_1(\bt_1+[s_1^1]+[s_2^1]+[s_3^1]-[z],\bt_2+[s_1^2]+[s_2^2]+[s_3^2]) \tau_2(\bt_1+[s_0^1]+[z],\bt_2+[s_0^2])\nn\\
	=&	\sum_{\mathrm{c.p.}\{s_1^1,  s_2^1, s_3^1\}}\frac{s_1^1(s_1^1-s_0^1)}{(s_1^1-s_2^1)(s_1^1-s_3^1)}\tau_1(\bt_1+[s_2^1]+[s_3^1],\bt_2+[s_1^2]+[s_2^2]+[s_3^2]) \cdot \nn\\
&\qquad \cdot \tau_2(\bt_1+[s_0^1]+[s_1^1],\bt_2+[s_0^2]).
\end{align*}
The right hand side of equation \eqref{ble-tau} is calculated in the same way. Thus the theorem is proved.
\end{prf}

\begin{cor}
The tau functions $\tau_1(\bt_1,\bt_2)$ and $\tau_2(\bt_1,\bt_2)$ of the KP-mKP hierarchy satisfy the following differential Fay identities:
\begin{align}
&\tau_{1}(\bt_1,\bt_2) \p_1\big(  \tau_{2}(\bt_1+[s_1^1],\bt_2-[s_1^2]) \big)-\p_1\big(\tau_{1}(\bt_1,\bt_2)\big)  \tau_{2}(\bt_1+[s_1^1],\bt_2-[s_1^2]) \nn\\
=&\frac{1}{s_1^1} \left(\tau_{1}(\bt_1,\bt_2) \tau_{2}(\bt_1+[s_1^1],\bt_2-[s_1^2]) - \tau_1(\bt_1+[s_1^1],\bt_2)\tau_2(\bt_1,\bt_2-[s_1^2]) \right) ,\label{fayt122}
\\
&\p_2\big( \tau_{1}(\bt_1-[s_1^1],\bt_2+[s_1^2])\big) \tau_{2}(\bt_1,\bt_2) - \tau_{1}(\bt_1-[s_1^1],\bt_2+[s_1^2])\p_2\big( \tau_{2}(\bt_1,\bt_2)\big) \nn\\
=&\frac{1}{s_1^2} \left(  \tau_{1}(\bt_1-[s_1^1],\bt_2+[s_1^2]) \tau_{2}(\bt_1,\bt_2) - \tau_2(\bt_1,\bt_2+[s_1^2])\tau_1(\bt_1-[s_1^1],\bt_2) \right). \label{fayt122b}
\end{align}
\end{cor}
\begin{prf}
Taking the derivative of the Fay identity \eqref{kpmkp-fay} with respect to $s_0^1$, and letting  $s_0^\nu,s_3^\nu\to0$ for $\nu\in\{1,2\}$, we have
\begin{align*}
&\frac{s_1^1}{s_1^1-s_2^1}\tau_1(\bt_1+[s_2^1],\bt_2+[s_1^2]+[s_2^2])\p_1\tau_2(\bt_1+[s_1^1],\bt_2) \\
&+
\frac{s_2^1}{s_2^1-s_1^1}\tau_1(\bt_1+[s_1^1],\bt_2+[s_1^2]+[s_2^2])\p_1\tau_2(\bt_1+[s_2^1],\bt_2)
\\
&-\frac{1}{s_1^1-s_2^1}\tau_1(\bt_1+[s_2^1],\bt_2+[s_1^2]+[s_2^2])\tau_2(\bt_1+[s_1^1],\bt_2) \\
&- \frac{1}{s_2^1-s_1^1}\tau_1(\bt_1+[s_1^1],\bt_2+[s_1^2]+[s_2^2])\tau_2(\bt_1+[s_2^1],\bt_2)
\\
=& \frac{s_1^2}{s_1^2-s_2^2}\tau_2(\bt_1+[s_1^1]+[s_2^1],\bt_2+[s_2^2])\p_1\tau_1(\bt_1,\bt_2+[s_1^2] ) \\
&+
\frac{s_2^2}{s_2^2-s_1^2}\tau_2(\bt_1+[s_1^1]+[s_2^1],\bt_2+[s_1^2])\p_1\tau_1(\bt_1,\bt_2+[s_2^2]).
\end{align*}
We let $s_2^1,s_2^2\to0$, and do shift $\bt_2\mapsto\bt_2-[s_1^2]$, then obtain the equality \eqref{fayt122}. The equality \eqref{fayt122b} is verified in the same way. The corollary is proved.
\end{prf}

\begin{lem}\label{thm-pwave}
For the (adjoint) Baker-Akhiezer functions and the tau functions of the KP-mKP hierarchy, the following equalities hold (recall $e^\beta=\tau_2(\bt_1,\bt_2)/\tau_1(\bt_1,\bt_2)$):
\begin{align}
\p_1\left(e^{-\beta}w_1(\bt_1,\bt_2;z) \right)=&z\frac{\tau_2(\bt_1-[z^{-1}],\bt_2)}{\tau_2(\bt_1,\bt_2)}e^{-\beta+\xi(\bt_1;z)},\label{pa1w1}\\
\p_1\left(e^{\beta}w_1^{\dag}(\bt_1,\bt_2;z) \right)=&-z\frac{\tau_1(\bt_1+[z^{-1}],\bt_2)}{\tau_1(\bt_1,\bt_2)}e^{\beta-\xi(\bt_1;z)},\label{pa1w1+}
\\
\p_2\left(w_2(\bt_1,\bt_2;z)  \right)=&z\frac{\tau_1(\bt_1,\bt_2-[z^{-1}])}{\tau_1(\bt_1,\bt_2)}e^{\beta+\xi(\bt_2;z)},\label{pa2w2}\\
\p_2\left(w_2^\dag(\bt_1,\bt_2;z)  \right)=&-z\frac{\tau_2(\bt_1,\bt_2+[z^{-1}])}{\tau_2(\bt_1,\bt_2)}e^{-\beta-\xi(\bt_2;z)}. \label{pa2w2+}
\end{align}
\end{lem}
\begin{prf}
By letting $s_1^2\to0$ and shifting $\bt_1\mapsto\bt_1-[s_1^1]$ in \eqref{fayt122}, one has
\[
-\p_1\left(\frac{\tau_{1}(\bt_1-[s_1^1],\bt_2)}{\tau_{2}(\bt_1,\bt_2)}\right)
=\frac{1}{s_1^1} \left(\frac{\tau_{1}(\bt_1-[s_1^1],\bt_2 )}{\tau_{2}(\bt_1,\bt_2)}- \frac{\tau_1(\bt_1,\bt_2)\tau_2(\bt_1-[s_1^1],\bt_2)}{\tau_2(\bt_1,\bt_2)^2}  \right).
\]
Then according to Theorem~\ref{thm-ble2}, we obtain
\begin{align*}
\p_1\left(w_1(\bt_1,\bt_2;z) \right)=&\p_1\left(\frac{\tau_1(\bt_1-[z^{-1}],\bt_2)}{\tau_2(\bt_1,\bt_2)} \frac{\tau_2(\bt_1,\bt_2)}{\tau_1(\bt_1,\bt_2)}e^{\xi(\bt_1;z)}\right)
\\
=&-z \left(\frac{\tau_{1}(\bt_1-[z^{-1}],\bt_2 )}{\tau_{2}(\bt_1,\bt_2)}- \frac{\tau_1(\bt_1,\bt_2)\tau_2(\bt_1-[z^{-1}],\bt_2)}{\tau_2(\bt_1,\bt_2)^2}\right) \frac{\tau_2(\bt_1,\bt_2)}{\tau_1(\bt_1,\bt_2)}e^{\xi(\bt_1;z)}
\\
&+\frac{\tau_1(\bt_1-[z^{-1}],\bt_2)}{\tau_2(\bt_1,\bt_2)} \p_1\left(\frac{\tau_2(\bt_1,\bt_2)}{\tau_1(\bt_1,\bt_2)}\right)e^{\xi(\bt_1;z)} +z w_1(\bt_1,\bt_2;z)
\\
=&z\frac{\tau_2(\bt_1-[z^{-1}],\bt_2)}{\tau_2(\bt_1,\bt_2)}e^{\xi(\bt_1;z)}+\p_1(\beta)w_1(\bt_1,\bt_2;z),
\end{align*}
which is equivalent to \eqref{pa1w1}.

With the same method, letting $s_1^1\to0$ and shifting $\bt_2\mapsto\bt_2-[s_1^2]$ in \eqref{fayt122b}, one has
\[-\p_2\left(\frac{\tau_2(\bt_1,\bt_2-[s_1^2])}{\tau_1(\bt_1,\bt_2)}\right) =\frac{1}{s_1^2}\left( \frac{\tau_2(\bt_1,\bt_2-[s_1^2])}{\tau_1(\bt_1,\bt_2)} -\frac{\tau_1(\bt_1,\bt_2-[s_1^2])\tau_2(\bt_1,\bt_2)}{\tau_1(\bt_1,\bt_2)^2} \right).\label{p2tau2tau1}
\]
Hence,
\begin{align*}
\p_2\left(w_2(\bt_1,\bt_2;z)  \right)=& \p_2\left(\frac{\tau_2(\bt_1-[z^{-1}],\bt_2)}{\tau_1(\bt_1,\bt_2)} e^{\xi(\bt_2;z)}\right)
\\
=&-z \left( \frac{\tau_2(\bt_1,\bt_2-[z^{-1}])}{\tau_1(\bt_1,\bt_2)} -\frac{\tau_1(\bt_1,\bt_2-[z^{-1}])\tau_2(\bt_1,\bt_2)}{\tau_1(\bt_1,\bt_2)^2} \right)e^{\xi(\bt_2;z)} \\
 &+z w_2(\bt_1,\bt_2;z)
\\
=&z\frac{\tau_1(\bt_1,\bt_2-[z^{-1}])}{\tau_1(\bt_1,\bt_2)}e^{\beta+\xi(\bt_2;z)}.
\end{align*}

The other two cases are similar. Thus the lemma is proved.
\end{prf}

Now we are ready to show the following result.

\begin{thm}\label{thm-tau12KP}
Each of the tau functions $\tau_1(\bt_1,\bt_2)$ and $\tau_2(\bt_1,\bt_2)$ of the KP-mKP hierarchy solves the KP hierarchy with time variables $\bt_1$ or with $\bt_2$.
\end{thm}
\begin{prf}
According to Theorem~\ref{thm-ble2}, we have
\begin{align*}
	\res_{z=\infty} z^{-1}\p_2\Big(w_2(\bt_1,\bt_2;z) \Big)w_2^\dag(\bt_1,\bt_2';z)\od z=0.
\end{align*}
By using \eqref{pa2w2} and \eqref{w2tau}, it lead to
\begin{align}
	\res_{z=\infty} \tau_1(\bt_1,\bt_2-[z^{-1}])\tau_1(\bt_1,\bt_2'+[z^{-1}])e^{\xi(\bt_2-\bt_2',z)}\od z=0,\label{KPtau1}
\end{align}
which is just the bilinear equation of the KP hierarchy with time variable $\bt_2$.

Similarly, let $\bt_1'=\bt_1$ and take the derivative with respect to $t_{2,1}'$ of the bilinear equation \eqref{exble3}, then with the help of \eqref{pa2w2+} one shows that $\tau_2(\bt_1,\bt_2)$ is a tau function of the KP hierarchy with time variable $\bt_2$.

In the same way, one can show that $\tau_1(\bt_1,\bt_2)$ is a tau function of the KP hierarchy with time variables $\bt_1$ or with $\bt_2$. Therefore the theorem is proved.
\end{prf}

It follows immediately that the tau functions $\tau_1$ and $\tau_2$ satisfy the Fay identity of the KP hierarchy (see Proposition 6.4.12 in \cite{Dickey}) and
the differential Fay identity of the KP hierarchy (see equation (3.11) in \cite{AvM}). More exactly, we arrive at the following result.

\begin{cor}
The tau functions $\tau_1(\bt_1,\bt_2)$ and $\tau_2(\bt_1,\bt_2)$ of the KP-mKP hierarchy satisfy
\begin{align}
\sum_{\mathrm{c.p.}\{s_1^1,s_2^1,s_3^1\}}(s_0^1-s_1^1)(s_2^1-s_3^1)\tau_\nu(\bt_1+[s_0^1]+[s_1^1],\bt_2)\tau_\nu(\bt_1+[s_2^1]+[s_3^1],\bt_2)=0,\label{fay-kp}
\end{align}
and
\begin{align}
&\p_1\left(\frac{\tau_\nu(\bt_1+[s_0^1]-[s_1^1],\bt_2)}{\tau_\nu(\bt_1,\bt_2)} \right) \nn\\
=&\left(\frac{1}{s_0^1}-\frac{1}{s_1^1}\right)\left( \frac{\tau_\nu(\bt_1+[s_0^1]-[s_1^1],\bt_2)}{\tau_\nu(\bt_1,\bt_2)} -\frac{\tau_\nu(\bt_1+[s_0^1],\bt_2)\tau_\nu(\bt_1-[s_1^1],\bt_2)}{\tau_\nu(\bt_1,\bt_2)^2}\right),\label{dfay-kp} \\
&
\p_2\left(\frac{\tau_\nu(\bt_1,\bt_2+[s_0^2]-[s_1^2])}{\tau_\nu(\bt_1,\bt_2)} \right)\nn \\
=&\left(\frac{1}{s_0^2}-\frac{1}{s_1^2}\right)\left( \frac{\tau_\nu(\bt_1,\bt_2+[s_0^2]-[s_1^2])}{\tau_\nu(\bt_1,\bt_2)} -\frac{\tau_\nu(\bt_1,\bt_2+[s_0^2])\tau_\nu(\bt_1,\bt_2-[s_1^2])}{\tau_\nu(\bt_1,\bt_2)^2}\right),\label{dfay-kpb}
\end{align}
where $\nu\in\{1,2\}$.
\end{cor}

At the end of this section, let us show two lemmas that will be employed below.
\begin{lem} \label{thm-taueq}
The tau functions $\tau_1(\bt_1,\bt_2)$ and $\tau_2(\bt_1,\bt_2)$ of the KP-mKP hierarchy satisfy the following equalities with pairwise distinct nonzero parameters $s_1$, $s_2$ and $s_3$:
	\begin{align}
&\tau_1(\bt_1+[s_{1}]-[s_{2}]-[s_{3}],\bt_2)\tau_1(\bt_1,\bt_2) \nn \\
=&\frac{s_2(s_3-s_1)}{s_1(s_3-s_2)}\tau_1(\bt_1+[s_{1}]-[s_{2}],\bt_2)\tau_1(\bt_1-[s_{3}],\bt_2)\nn\\
	& +\frac{s_3(s_2-s_1)}{s_1(s_2-s_3)}\tau_1(\bt_1+[s_{1}]-[s_{3}],\bt_2)\tau_1(\bt_1-[s_{2}],\bt_2), \label{kp-fay1}
\\
&\tau_1(\bt_1,\bt_2)\tau_2(\bt_1+[s_{1}]-[s_{2}],\bt_2-[s_{3}]) \nn\\
=& \left(1-\frac{s_2}{s_1}\right)\tau_1(\bt_1-[s_{2}],\bt_2)\tau_2(\bt_1+[s_{1}],\bt_2-[s_{3}])\nn\\
 &+\frac{s_2}{s_1}\tau_1(\bt_1+[s_{1}]-[s_{2}],\bt_2)\tau_2(\bt_1,\bt_2-[s_{3}]), \label{tau21}
 \\
 &\tau_1(\bt_1-[s_{3}],\bt_2+[s_{1}]-[s_{2}])\tau_2(\bt_1,\bt_2) \nn\\
 =& \left(1-\frac{s_2}{s_1}\right) \tau_1(\bt_1-[s_{3}],\bt_2+[s_{1}])\tau_2(\bt_1,\bt_2-[s_{2}])\nn\\
			&+\frac{s_2}{s_1}\tau_1(\bt_1-[s_{3}],\bt_2)\tau_2(\bt_1,\bt_2+[s_{1}]-[s_{2}]), \label{tau12}
 \\
 &\tau_1(\bt_1,\bt_2)\tau_2(\bt_1,\bt_2+[s_{1}]-[s_{2}]-[s_{3}]) \nn\\
=&\frac{s_{3}-s_{1}}{s_{3}-s_{2}}\tau_1(\bt_1,\bt_2+[s_{1}]-[s_{2}])\tau_2(\bt_1,\bt_2-[s_{3}])\nn\\
&+\frac{s_{2}-s_{1}}{s_{2}-s_{3}}\tau_1(\bt_1,\bt_2+[s_{1}]-[s_{3}])\tau_2(\bt_1,\bt_2-[s_{2}]).\label{tau121}
	\end{align}
\end{lem}
\begin{prf}
The verification is straightforward. Firstly, let us take $\nu=1$ in \eqref{fay-kp}. Let $s_0^1\to0$ and do replacement
\[
\bt_1\mapsto\bt_1-[s_2^1]-[s_3^1],
\]
then the equality \eqref{kp-fay1} follows by replacing $s_i^1\mapsto s_i$.

Secondly, in the Fay identity \eqref{kpmkp-fay} we let $s_0^1,\, s_0^2,\, s_3^1\to0$ and $s_2^2=-s_3^2\to0$, then we have
\begin{align*}
&\frac{s_1^1}{s_1^1-s_2^1}\tau_1(\bt_1+[s_2^1],\bt_2+[s_1^2])\tau_2(\bt_1+[s_1^1],\bt_2)  \nn\\
&+\frac{s_2^1}{s_2^1-s_1^1}\tau_1(\bt_1+[s_1^1],\bt_2+[s_1^2])\tau_2(\bt_1+[s_2^1],\bt_2)  \nn\\
=& \tau_1(\bt_1,\bt_2+[s_1^2])\tau_2(\bt_1+[s_1^1]+[s_2^1],\bt_2).
\end{align*}
By replacements $\bt_1\mapsto \bt_1-[s_2^1]$ and  $\bt_2\mapsto \bt_2-[s_1^2]$, one derives the equality \eqref{tau21} after renaming parameters.

The equality \eqref{tau12} is derived in the same way due to the form of the Fay identity  \eqref{kpmkp-fay}.

Finally, in the   Fay identity  \eqref{kpmkp-fay} one lets $s_0^1,\, s_3^1,\,s_3^2\to0$, and lets $s_1^1=-s_2^1\to0$, then obtains
\begin{align*}
&\tau_1(\bt_1,\bt_2+[s_1^2]+[s_2^2])\tau_2(\bt_1,\bt_2+[s_0^2]) \\
 =& \frac{s_1^2-s_0^2}{s_1^2-s_2^2} \tau_1(\bt_1,\bt_2+[s_0^2]+[s_1^2])\tau_2(\bt_1,\bt_2+[s_2^2]) \\
 & +\frac{s_2^2-s_0^2}{s_2^2-s_1^2} \tau_1(\bt_1,\bt_2+[s_0^2]+[s_2^2])\tau_2(\bt_1,\bt_2+[s_1^2]).
\end{align*}
Doing replacement $\bt_2\mapsto \bt_2-[s_1^2]-[s_2^2]$, one confirms the equality  \eqref{tau121} with parameters renamed. The lemma is proved.
\end{prf}

\begin{lem}\label{thm-pwave1}
The (adjoint) Baker-Akhiezer functions $w_\nu(z)=w_\nu(\bt_1, \bt_2; z)$ and $w_\nu^\dag(z)=w_\nu^\dag(\bt_1, \bt_2; z)$ of the KP-mKP hierarchy satisfy the following equalities:
\begin{align}
e^{-\beta}\p_1e^{\beta}w_1^\dag(\gm)\cdot w_1(z)=&\frac{\gm}{\gm-z}\p_1\left(e^{-\xi(\bt_1;\gm)+\xi(\bt_1;z)} \frac{\tau_1(\bt_1+[\gm^{-1}]-[z^{-1}],\bt_2)}{\tau_1(\bt_1,\bt_2)}\right),\label{pwave1}\\
e^{-\beta}\p_1e^{\beta}w_1^\dag(\gm)\cdot w_2(z)=&\p_1\left(e^{-\xi(\bt_1;\gm)+\xi(\bt_2;z)}\frac{\tau_2(\bt_1+[\gm^{-1}],\bt_2-[z^{-1}])}{\tau_1(\bt_1,\bt_2)} \right),\label{pwave2}\\
w_2^\dag(\gm)\p_2w_1(z)=&\p_2\left( e^{-\xi(\bt_2;\gm)+\xi(\bt_1;z)}
\left(-\frac{\tau_1(\bt_1-[z^{-1}],\bt_2+[\gm^{-1}])}{\tau_2(\bt_1,\bt_2)} \right.\right.\nn\\ &\quad\quad+\left.\left.\frac{\tau_1(\bt_1-[z^{-1}],\bt_2) \tau_1(\bt_1,\bt_2+[\gm^{-1}])}{\tau_1(\bt_1,\bt_2) \tau_2(\bt_1,\bt_2)}\right) \right),
\label{pwave4} \\
w_2^\dag(\gm)\p_2w_2(z)=&\frac{z}{z- \gm }\p_2\left(e^{ -\xi(\bt_2;\gm)+\xi(\bt_2;z) }\frac{\tau_1(\bt_1,\bt_2+[\gm^{-1}]-[z^{-1}])}{\tau_{1}(\bt_1,\bt_2)} \right). \label{pwave3}
\end{align}
\end{lem}
\begin{proof}Let us check the above equalities case by case.
\\
(i) By using the differential Fay identity \eqref{dfay-kp} with $\nu=1$, for \eqref{pwave1} one has
\begin{align}
	\hbox{r.h.s.}=&-\gm e^{\xi(\bt_1;z)-\xi(\bt_1;\gm)}\frac{\tau_1(\bt_1+[\gm^{-1}]-[z^{-1}],\bt_2)}{\tau_1(\bt_1,\bt_2)}\nn\\
	&+\frac{\gm}{\gm-z} e^{\xi(\bt_1;z)-\xi(\bt_1;\gm)}(\gm-z)\left( \frac{\tau_1(\bt_1+[\gm^{-1}]-[z^{-1}],\bt_2)}{\tau_1(\bt_1,\bt_2)} \right. \nn\\
	& \left.  -\frac{\tau_1(\bt_1+[\gm^{-1}],\bt_2)\tau_1(\bt_1-[z^{-1}],\bt_2)}{\tau_1(\bt_1,\bt_2)^2}\right)\nn\\
	=&-\gm e^{\xi(\bt_1;z)-\xi(\bt_1;\gm)}\frac{\tau_1(\bt_1+[\gm^{-1}],\bt_2) \tau_1(\bt_1-[z^{-1}],\bt_2)}{\tau_1(\bt_1,\bt_2)^2}\nn \\
=& e^{-\beta}\p_1e^{\beta}(w_1^\dag(\gm))\cdot w_1(z)=\hbox{l.h.s.}, \nn
 \end{align}
where the third equality is due to \eqref{pa1w1+} and \eqref{w1tau}.
\\
(ii) Similarly, by using \eqref{fayt122}, for \eqref{pwave2} one has
\begin{align*} \hbox{r.h.s.}=& -\gm e^{-\xi(\bt_1;\gm)+\xi(\bt_2;z)}\frac{\tau_2(\bt_1+[\gm^{-1}],\bt_2-[z^{-1}])}{\tau_1(\bt_1,\bt_2)}\nn\\
&+e^{-\xi(\bt_1;\gm)+\xi(\bt_2;z)}\gm\left( \frac{\tau_2(\bt_1+[\gm^{-1}],\bt_2-[z^{-1}])}{\tau_1(\bt_1,\bt_2)}
-\frac{\tau_1(\bt_1+[\gm^{-1}],\bt_2) \tau_2(\bt_1,\bt_2-[z^{-1}])}{\tau_1(\bt_1,\bt_2)^2}\right)\nn\\
=&-\gm e^{-\xi(\bt_1;\gm)+\xi(\bt_2;z)}\frac{\tau_1(\bt_1+[\gm^{-1}],\bt_2)\tau_2(\bt_1,\bt_2-[z^{-1}])}{\tau_1(\bt_1,\bt_2)^2} \\
=&e^{-\beta}\p_1e^{\beta}(w_1^\dag(\gm))w_2(z)=\hbox{l.h.s.},
\end{align*}
where the third equality is due to \eqref{pa1w1+} and \eqref{w2tau}.
\\
(iii) By using \eqref{fayt122b}, for \eqref{pwave4} we have
\begin{align}
\hbox{r.h.s.}
=&e^{-\xi(\bt_2;\gm)+\xi(\bt_1;z)}\left(
\gm \frac{\tau_1(\bt_1-[z^{-1}],\bt_2+[\gm^{-1}])}{\tau_2(\bt_1,\bt_2)} - \gm   \frac{\tau_1(\bt_1-[z^{-1}],\bt_2)\tau_1(\bt_1,\bt_2+[\gm^{-1}])}{\tau_1(\bt_1,\bt_2)\tau_2(\bt_1,\bt_2)}\right.
 \nn\\
&-\gm\left(  \frac{\tau_1(\bt_1-[z^{-1}],\bt_2+[\gm^{-1}])}{\tau_2(\bt_1,\bt_2)} - \frac{\tau_1(\bt_1-[z^{-1}],\bt_2)\tau_2(\bt_1,\bt_2+[\gm^{-1}])}{\tau_2(\bt_1,\bt_2)^2}  \right)
 \nn\\
&+\p_2\left(\frac{\tau_1(\bt_1-[z^{-1}],\bt_2)}{\tau_1(\bt_1,\bt_2)}\right) \frac{\tau_1(\bt_1,\bt_2+[\gm^{-1}])}{\tau_2(\bt_1,\bt_2)}
\nn\\
&\left.+\frac{\tau_1(\bt_1-[z^{-1}],\bt_2)}{\tau_1(\bt_1,\bt_2)}\cdot\gm \left(  \frac{\tau_1(\bt_1,\bt_2+[\gm^{-1}])}{\tau_2(\bt_1,\bt_2)} - \frac{\tau_1(\bt_1,\bt_2)\tau_2(\bt_1,\bt_2+[\gm^{-1}])}{\tau_2(\bt_1,\bt_2)^2}  \right) \right)
\nn\\
=&\frac{\tau_1(\bt_1+[\gm^{-1}],\bt_2)}{\tau_2(\bt_1,\bt_2)}e^{-\xi(\bt_2;\gm)+\xi(\bt_1;z)} \p_2\left(\frac{\tau_1(\bt_1-[z^{-1}],\bt_2)}{\tau_1(\bt_1,\bt_2)}\right)
\nn\\
=&w_2^\dag(\gm)\p_2(w_1(z))=\hbox{l.h.s.} \nn
\end{align}
\\
(iv)
For  \eqref{pwave3}, by using \eqref{dfay-kpb} with $\nu=1$ one has
\begin{align*}
\hbox{r.h.s.}
=&z	e^{-\xi(\bt_2;\gm)+\xi(\bt_2;z)} \frac{\tau_1(\bt_1,\bt_2+[\gm^{-1}]-[z^{-1}])}{\tau_{1}(\bt_1,\bt_2)} \nn\\
&+\frac{z}{z-\gm}e^{-\xi(\bt_2;\gm)+\xi(\bt_2;z)}(\gm-z) \cdot
\\
&\cdot\left( \frac{\tau_1(\bt_1,\bt_2+[\gm^{-1}]-[z^{-1}])}{\tau_{1}(\bt_1,\bt_2)} - \frac{\tau_1(\bt_1,\bt_2+[\gm^{-1}])\tau_1(\bt_1,\bt_2-[z^{-1}])}{\tau_{1}(\bt_1,\bt_2)^2} \right)
\nn\\
=&z e^{-\xi(\bt_2;\gm)+\xi(\bt_2;z)}\frac{\tau_1(\bt_1,\bt_2+[\gm^{-1}])\tau_1(\bt_1,\bt_2-[z^{-1}])}{\tau_{1}(\bt_1,\bt_2)^2} \\
=& w_2^\dag(\gm)\p_2(w_2(z))=\hbox{l.h.s.},
\end{align*}
where the third equality is due to \eqref{pa2w2} and \eqref{w2tau}.
 Therefore the lemma is proved.
\end{proof}

\section{Additional symmetries}
Now we are to construct a series of additional symmetries of the KP-mKP hierarchy. To start with, let us introduce two Orlov-Schulman operators as follows:
\begin{align}
M_\nu=\Phi_\nu\Gm_\nu\Phi_\nu^{-1},\quad \nu\in\{1,2\},\label{Mmu}
\end{align}
where $\Phi_\nu$ are the dressing operators given in \eqref{re-Phi120} and
\[\Gm_\nu=\sum_{k=1}^{\infty}kt_{\nu,k}\p_\nu^{k-1}. \]
Recall the pseudo-differential operators $P_\nu$ in \eqref{exP2} and the Baker-Akhiezer functions $w_\nu$ in \eqref{wavef}, clearly one has
\begin{align}
	[P_\nu,\, M_\nu]=1,\quad \p_zw_\nu=M_\nu w_\nu,\label{PM}
\end{align}
where $\p_z=\frac{\p}{\p z}$. Moreover, with the help of \eqref{Phinutmu} and \eqref{Mmu}, it is easy to see that
\begin{align} \label{Mt}
\frac{\p M_\nu}{\p t_{\nu,k}}=[(P_\nu^k)_{\geq\nu-1},\, M_\nu],\quad \frac{\p M_\nu}{\p t_{\mu,k}}=[(P_\mu^k)_{\geq\mu-1}\la\Phi_\nu\ra\Phi_\nu^{-1},\, M_\nu],
\end{align}
where $\nu,\mu\in\{1,2\}$ and $\nu\neq\mu$.

For $\nu,\mu\in\{1,2\}$, $m\in\Z_{\geq0}$ and $l\in\Z$, denote
\begin{align}\label{add-A}
	A_{\nu,ml}^\mu=\left\{\begin{array}{cl}
		-(M_\nu^mP_\nu^l)_{<\nu-1}, & \quad \nu=\mu; \\
		(M_\nu^mP_\nu^l)_{\geq\nu-1}\la\Phi_\mu\ra\Phi_\mu^{-1},& \quad \nu\neq\mu.
	\end{array}
	\right.
\end{align}
It is easy to see that the following evolutionary equations are well defined:
\begin{align}
	\frac{\p \Phi_\mu}{\p s_{m,l}^{(\nu)}}=A_{\nu,ml}^\mu\Phi_\mu.\label{addde}
\end{align}

\begin{lem}\label{thm-wpm}
For $\nu,\mu\in\{1,2\}$, $m,m'\in\Z_{\geq0}$ and $l,l'\in\Z$, the following equalities hold:
\begin{align}
\frac{\p w_\mu}{\p s_{m,l}^{(\nu)}}=&A_{\nu,ml}^\mu w_\mu,\label{add1}\\
\frac{\p P_\mu}{\p s_{m,l}^{(\nu)}}=&\Big[A_{\nu,ml}^\mu,\, P_\mu\Big],\label{add2}\\
\frac{\p M_\mu}{\p s_{m,l}^{(\nu)}}=&\Big[A_{\nu,ml}^\mu,\, M_\mu\Big],\label{add3}\\
 \frac{\p M_\mu^{m'}P_\mu^{l'}}{\p s_{m,l}^{(\nu)}} =&\Big[A_{\nu,ml}^\mu,\, M_\mu^{m'}P_\mu^{l'}\Big].\label{add4}
\end{align}
\end{lem}
\begin{prf}
 According to the definition \eqref{addde}, it is easy to see \eqref{add1}--\eqref{add3} by using \eqref{wavef},   \eqref{exP2} and \eqref{Mmu}. Clearly the equalities  \eqref{add2}--\eqref{add3} lead to \eqref{add4}. Thus the lemma is proved.
\end{prf}

\begin{prp}\label{thm-com}
The flows \eqref{addde} commute with those \eqref{Phinutmu} of the KP-mKP hierarchy. More exactly, for $\nu,\mu,\ld\in\{1,2\}$, the following equalities hold:
\begin{align}
\Big[\frac{\p}{\p t_{\mu,k}},\,  \frac{\p  }{\p s_{m,l}^{(\nu)}}\Big]\Phi_{\ld}=0,\quad k\in\Z_{\ge1},\, m\in\Z_{\geq0},\,l\in\Z.\label{tkadd}
\end{align}
\end{prp}
\begin{prf}
According to \eqref{Phinutmu} and \eqref{addde}, it is sufficient to verify the zero-curvature condition:
\begin{align}
\frac{\p A_{\nu,ml}^\ld }{\p t_{\mu,k}}-\frac{\p B_{\mu,k}^\ld }{\p s_{m,l}^{(\nu)}}+\Big[A_{\nu,ml}^\ld,\,B_{\mu,k}^\ld \Big]=0, \quad \nu,\mu,\ld\in\{1,2\}. \label{zero-cur}
\end{align}
Let us do it case by case, with the help of \eqref{Mt}, \eqref{add2} and (recall \eqref{Phinutmu})
\[
	\frac{\p P_\mu}{\p t_{\nu,k}}=\left[B_{\nu,k}^\mu, P_\mu\right], \quad \mu,\nu\in\set{1,2}, ~ k\in\Z_{\geq1}.
\]
\\
(i) When $\mu=\nu=\ld$, one has
\begin{align*}
&\frac{\p A_{\nu,ml}^\nu }{\p t_{\nu,k}}-\frac{\p B_{\nu,k}^\nu }{\p s_{m,l}^{(\nu)}}+\Big[A_{\nu,ml}^\nu,\,B_{\nu,k}^\nu \Big]\\
=&\Big[(P_\nu^k)_{\geq\nu-1},\, -M_\nu^mP_\nu^l \Big]_{<\nu-1}-	\Big[-(M_\nu^mP_\nu^l)_{<\nu-1},\, -P_\nu^k \Big]_{<\nu-1} \nn\\
&+\Big[-(M_\nu^mP_\nu^l)_{<\nu-1},\, -(P_\nu^k)_{<\nu-1}\Big]\\
=&\Big[(P_\nu^k)_{\geq\nu-1},\, -M_\nu^mP_\nu^l \Big]_{<\nu-1}+\Big[(P_\nu^k)_{\geq\nu-1},\,(M_\nu^mP_\nu^l)_{<\nu-1} \Big]_{<\nu-1}\\
=&\Big[(P_\nu^k)_{\geq\nu-1},\, -(M_\nu^mP_\nu^l)_{\geq\nu-1} \Big]_{<\nu-1}=0.
\end{align*}
\\
(ii) When $\mu=\nu\neq\ld$, we have
\begin{align*}
&\frac{\p A_{\nu,ml}^\ld }{\p t_{\nu,k}}-\frac{\p B_{\nu,k}^\ld }{\p s_{m,l}^{(\nu)}}+\Big[A_{\nu,ml}^\ld,\,B_{\nu,k}^\ld \Big]\nn\\
=&\Big[(P_\nu^k)_{\geq\nu-1},\, M_\nu^mP_\nu^l \Big]_{\geq\nu-1}\la \Phi_\ld\ra\Phi_\ld^{-1} +(M_\nu^mP_\nu^l)_{\geq\nu-1}\la(P_\nu^k)_{\geq\nu-1}\la\Phi_\ld\ra\ra\Phi_\ld^{-1}\nn\\
&-(M_\nu^mP_\nu^l)_{\geq\nu-1}\la\Phi_\ld\ra\Phi_\ld^{-1}  (P_\nu^k)_{\geq\nu-1}\la\Phi_\ld\ra\Phi_\ld^{-1}
-\Big[-(M_\nu^mP_\nu^l)_{<\nu-1},\,P_\nu^k \Big]_{\geq\nu-1}\la\Phi_\ld\ra\Phi_\ld^{-1}\nn\\
&-(P_\nu^k)_{\geq\nu-1}\la(M_\nu^mP_\nu^l)_{\geq\nu-1}\la\Phi_\ld\ra\ra\Phi_\ld^{-1}
+ (P_\nu^k)_{\geq\nu-1}\la\Phi_\ld\ra\Phi_\ld^{-1}
(M_\nu^mP_\nu^l)_{\geq\nu-1}\la\Phi_\ld\ra\Phi_\ld^{-1}\nn\\
&+\Big[ (M_\nu^mP_\nu^l)_{\geq\nu-1}\la\Phi_\ld\ra\Phi_\ld^{-1},\,(P_\nu^k)_{\geq\nu-1}\la\Phi_\ld\ra\Phi_\ld^{-1} \Big]\nn\\
=&\Big[(P_\nu^k)_{\geq\nu-1},\, M_\nu^mP_\nu^l \Big]_{\geq\nu-1}\la \Phi_\ld\ra\Phi_\ld^{-1}
+\Big[(M_\nu^mP_\nu^l)_{\geq\nu-1},\,  (P_\nu^k)_{\geq\nu-1}\Big]\la \Phi_\ld\ra\Phi_\ld^{-1}\nn\\
&+\Big[(M_\nu^mP_\nu^l)_{<\nu-1},\,P_\nu^k \Big]_{\geq\nu-1}\la\Phi_\ld\ra\Phi_\ld^{-1}
=0.
\end{align*}
\\
(iii) For $\mu\neq\nu=\ld$, we have
\begin{align}
&\frac{\p A_{\nu,ml}^\nu }{\p t_{\mu,k}}-\frac{\p B_{\mu,k}^\nu }{\p s_{m,l}^{(\nu)}}+\Big[A_{\nu,ml}^\nu,\,B_{\mu,k}^\nu \Big]\nn\\
=&\Big[(P_\mu^k)_{\geq\mu-1}\la\Phi_\nu\ra\Phi_\nu^{-1},\,-M_\nu^mP_\nu^l \Big]_{<\nu-1}
-\Big[(M_\nu^mP_\nu^l)_{\geq\nu-1}\la\Phi_\mu\ra\Phi_\mu^{-1},\,P_\mu^k \Big]_{\geq\mu-1}\la\Phi_\nu\ra\Phi_\nu^{-1}\nn\\
&-(P_\mu^k)_{\geq\mu-1}\la - (M_\nu^mP_\nu^l)_{<\nu-1}\Phi_\nu\ra\Phi_\nu^{-1}
+(P_\mu^k)_{\geq\mu-1}\la\Phi_\nu\ra\Phi_\nu^{-1}(- (M_\nu^mP_\nu^l)_{<\nu-1})\nn\\
&+\Big[- (M_\nu^mP_\nu^l)_{<\nu-1},\, (P_\mu^k)_{\geq\mu-1}\la\Phi_\nu\ra\Phi_\nu^{-1}\Big]\nn\\
=&\Big[\res_{\p_\mu}(P_\mu^k)_{\geq\mu-1}\Phi_\nu\p_\mu^{-1}\Phi_\nu^{-1},\,-M_\nu^mP_\nu^l \Big]_{<\nu-1}
-\Big[(M_\nu^mP_\nu^l)_{\geq\nu-1}\la\Phi_\mu\ra\Phi_\mu^{-1},\,P_\mu^k \Big]_{\geq\mu-1}\la\Phi_\nu\ra\Phi_\nu^{-1}\nn\\
&+\Big(\res_{\p_\mu}\Big[(P_\mu^k)_{\geq\mu-1},\, (M_\nu^mP_\nu^l)_{<\nu-1} \Big]\Phi_\nu\p_\mu^{-1}\Phi_\nu^{-1}\Big)_{<\nu-1}\nn\\
=&\Big(\res_{\p_\mu}\Big[(P_\mu^k)_{\geq\mu-1},\, -(M_\nu^mP_\nu^l)_{\geq\nu-1} \Big]\Phi_\nu\p_\mu^{-1}\Phi_\nu^{-1}\Big)_{<\nu-1}\nn\\
&-\Big[(M_\nu^mP_\nu^l)_{\geq\nu-1}\la\Phi_\mu\ra\Phi_\mu^{-1},\,P_\mu^k \Big]_{\geq\mu-1}\la\Phi_\nu\ra\Phi_\nu^{-1}, \label{add-T0}
\end{align}
where the last equality holds for the fact
\begin{align*}
\Phi_\nu\Big[\p_\mu^{-1},\, \Gm_\nu^m\p_\nu^l \Big]\Phi_\nu^{-1}=0.
\end{align*}
On the one hand, it is easy to see that
\begin{align}
\Big(\Big[(M_\nu^mP_\nu^l)_{\geq\nu-1}\la\Phi_\mu\ra\Phi_\mu^{-1},\,P_\mu^k \Big]_{\geq\mu-1}\la\Phi_\nu\ra\Phi_\nu^{-1}\Big)_{\geq\nu-1}=0. \label{add-T01}
\end{align}
On the other hand, since
\begin{align}\label{pkphi}
	\p_\nu^n\la\Phi_\mu\ra=\p_\nu^n\Phi_\mu-\sum_{r=1}^n\binom{n}{r}\p_{\nu}^{n-r}\la\Phi_\mu\ra\p_{\nu}^r, \quad n\ge1,
\end{align}
then there exists an operator $T=\sum_{i\geq0}\sum_{j\leq0}T_{ij}\p_\nu^i\p_\mu^j$ ($T_{ij}\in\cB$) such that
\begin{align*}
(M_\nu^mP_\nu^l)_{\geq\nu-1}\la\Phi_\mu\ra\Phi_\mu^{-1}=(M_\nu^mP_\nu^l)_{\geq\nu-1}+T\p_\nu\Phi_\mu^{-1}.
\end{align*}
Hence, one has
\begin{align}
&\Big(\Big[(M_\nu^mP_\nu^l)_{\geq\nu-1}\la\Phi_\mu\ra\Phi_\mu^{-1},\,P_\mu^k \Big]_{\geq\mu-1}\la\Phi_\nu\ra\Phi_\nu^{-1}\Big)_{<\nu-1}\nn\\
=&\Big(\res_{\p_\mu}\Big[(M_\nu^mP_\nu^l)_{\geq\nu-1}\la\Phi_\mu\ra\Phi_\mu^{-1},\,P_\mu^k \Big] \Phi_\nu\p_\mu^{-1}\Phi_\nu^{-1}\Big)_{<\nu-1}\nn\\
&-\delta_{\mu,2}\Big(\res_{\p_\mu}\Big[(M_\nu^mP_\nu^l)_{\geq\nu-1}\la\Phi_\mu\ra\Phi_\mu^{-1},\,P_\mu^k \Big] \p_\mu^{-1}\Big)_{<\nu-1}\nn\\
=&\Big(\res_{\p_\mu}\Big[(M_\nu^mP_\nu^l)_{\geq\nu-1}+T\p_\nu\Phi_\mu^{-1},\,P_\mu^k \Big] \Phi_\nu\p_\mu^{-1}\Phi_\nu^{-1}\Big)_{<\nu-1}-0\nn\\
=&\Big(\res_{\p_\mu}\Big[(M_\nu^mP_\nu^l)_{\geq\nu-1},\,(P_\mu^k)_{\geq\mu-1} \Big] \Phi_\nu\p_\mu^{-1}\Phi_\nu^{-1}\Big)_{<\nu-1}\nn\\
&+\delta_{\mu,2}\Big(\res_{\p_\mu}\Big[(M_\nu^mP_\nu^l)_{\geq\nu-1},\,(P_\mu^k)_{<\mu-1}\p_\mu^{-1} \Big] \Big)_{<\nu-1}\nn\\
&+\Big(\res_{\p_\mu}\Big[T\p_\nu\Phi_\mu^{-1},\,P_\mu^k \Big] \Phi_\nu\p_\mu^{-1}\Phi_\nu^{-1}\Big)_{<\nu-1}\nn\\
=&\Big(\res_{\p_\mu}\Big[(M_\nu^mP_\nu^l)_{\geq\nu-1},\,(P_\mu^k)_{\geq\mu-1} \Big] \Phi_\nu\p_\mu^{-1}\Phi_\nu^{-1}\Big)_{<\nu-1}\nn\\
&+\Big(\res_{\p_\mu}\Big[T\p_\nu\Phi_\mu^{-1},\,P_\mu^k \Big] \Phi_\nu\p_\mu^{-1}\Phi_\nu^{-1}\Big)_{<\nu-1}.\label{add-T}
\end{align}
According to the form of $T$ and \eqref{Phinumu}, one sees
\begin{align}
\Big(\res_{\p_\mu}\Big[T\p_\nu\Phi_\mu^{-1},\,P_\mu^k \Big] \Phi_\nu\p_\mu^{-1}\Phi_\nu^{-1}\Big)_{<\nu-1}=0.\label{add-T1}
\end{align}
In combination of \eqref{add-T0}--\eqref{add-T1}, we confirm the condition \eqref{zero-cur}.

Finally, the case $\mu=\lambda\ne\nu$ is similar. Therefore the proposition is proved.
\end{prf}

The equalities in \eqref{tkadd} mean that the flows \eqref{addde} are symmetries of the KP-mKP hierarchy, which are called the \emph{additional symmetries}.

For $\nu\in\{1,2\}$, since $[P_\nu,\, M_\nu]=1$, we see that each commutator $[M_\nu^{m}P_\nu^{l},\, M_\nu^{m'}P_\nu^{l'}]$ can be written as  a polynomial in $M_\nu$ and $P_\nu$. Namely, there exist constants $c_{ml,m'l'}^{qr}$ satisfying
\begin{align}
[M_\nu^{m}P_\nu^{l},\, M_\nu^{m'}P_\nu^{l'}]=\sum_{q,r}c_{ml,m'l'}^{qr}M_\nu^q P_\nu^r.\label{c}
\end{align}
For example,
\begin{align*}
c_{ml,m'l'}^{qr}=&0 \quad \text{whenever}\quad q\geq m+m' \, \text{or}\, |r-(l+l')|>\text{max}(m,m'),\\
c_{0l,0l'}^{qr}=&0,\quad c_{0l,1l'}^{qr}=l\delta_{q,0}\delta_{r,l+l'-1},\quad c_{1l,1l'}^{qr}=(l-l')\delta_{q,1}\delta_{r,l+l'-1}.
\end{align*}
We remark that such constants $c_{ml,m'l'}^{qr}$ are free of the index $\nu$.

\begin{thm}\label{thm-addc}
The flows of the additional symmetries of the KP-mKP hierarchy satisfy the following commutation relations:
\begin{align}
	\left[\frac{\p  }{\p s_{m,l}^{(\nu)}} ,\,  \frac{\p  }{\p s_{m',l'}^{(\nu)}}\right]\Phi_\ld =-\sum_{q,r}c_{ml,m'l'}^{qr}\frac{\p \Phi_\ld }{\p s_{q,r}^{(\nu)}},\quad
	\left[\frac{\p  }{\p s_{m,l}^{(\nu)}} ,\,  \frac{\p  }{\p s_{m',l'}^{(\mu)}}\right]\Phi_\ld =0,
\end{align}
where $\nu,\mu,\ld\in\{1,2\}$ with $\nu\neq\mu$, $m,m'\in\Z_{\ge0}$ and $l,l'\in\Z$. In other words, such flows give a realization of the
$w_{\infty}\times w_{\infty}$-algebra.
\end{thm}
\begin{prf}
The calculation is straightforward, with the help of \eqref{addde}, \eqref{add4} and \eqref{c}. For instance, when $\nu\neq\mu=\ld$, one has
\begin{align}
\left[\frac{\p  }{\p s_{m,l}^{(\nu)}},\,  \frac{\p  }{\p s_{m',l'}^{(\mu)}}\right]\Phi_\mu=&
	\frac{\p  }{\p s_{m,l}^{(\nu)}}\Big(  	A_{\mu,m'l'}^\mu\Phi_\mu\Big)-\frac{\p  }{\p s_{m',l'}^{(\mu)}}\Big(  	(M_\nu^{m}P_\nu^{l})_{\geq\nu-1}\la\Phi_\mu\ra\Big)\nn\\	
	=&\left[A_{\nu,ml}^\mu,\, -M_\mu^{m'}P_\mu^{l'} \right]_{<\mu-1}\Phi_\mu+A_{\mu,m'l'}^\mu A_{\nu,ml}^\mu\Phi_\mu\nn\\
	&-\left[A_{\mu,m'l'}^\nu ,\, M_\nu^mP_\nu^l\right]_{\geq\nu-1}\la\Phi_\mu\ra-(M_\nu^mP_\nu^l)_{\geq\nu-1}\la A_{\mu,m'l'}^{\mu}\Phi_\mu\ra\nn\\
	=&\bigg( \Big(\res_{\p_\nu}\left[(M_\nu^mP_\nu^l)_{\geq\nu-1}\Phi_\mu\p_\nu^{-1}\Phi_\mu^{-1},\,  -M_\mu^{m'}P_\mu^{l'} \right]\Big)_{<\mu-1}\nn\\
	&-\left[A_{\mu,m'l'}^\nu ,\, M_\nu^mP_\nu^l\right]_{\geq\nu-1}\la\Phi_\mu\ra\Phi_\mu^{-1}\nn\\
	&+\Big(\res_{\p_\nu}\left[(M_\nu^mP_\nu^l)_{\geq\nu-1},\, (M_\mu^{m'}P_\mu^{l'})_{<\mu-1} \right]\Phi_\mu\p_\nu^{-1}\Phi_\mu^{-1}\Big)_{<\mu-1}\bigg)\Phi_\mu\nn\\
	=&\bigg(-\left[A_{\mu,m'l'}^\nu ,\, M_\nu^mP_\nu^l\right]_{\geq\nu-1}\la\Phi_\mu\ra\Phi_\mu^{-1}\nn\\
	&+\Big(\res_{\p_\nu}\left[(M_\nu^mP_\nu^l)_{\geq\nu-1},\, -(M_\mu^{m'}P_\mu^{l'})_{\geq\mu-1} \right]\Phi_\mu\p_\nu^{-1}\Phi_\mu^{-1}\Big)_{<\mu-1} \bigg)\Phi_\mu.\label{addch}
\end{align}
One observes that \eqref{addch} is a pseudo-differential operator containing a single derivation $\p_\mu$. Clearly, one has \begin{align}
	&\bigg(-\left[A_{\mu,m'l'}^\nu ,\, M_\nu^mP_\nu^l\right]_{\geq\nu-1}\la\Phi_\mu\ra\Phi_\mu^{-1}\nn\\
	&+\Big(\res_{\p_\nu}\left[(M_\nu^mP_\nu^l)_{\geq\nu-1},\, -(M_\mu^{m'}P_\mu^{l'})_{\geq\mu-1} \right]\Phi_\mu\p_\nu^{-1}\Phi_\mu^{-1}\Big)_{<\mu-1} \bigg)_{\geq\mu-1}=0.\label{add-geq1}
\end{align}
What is more, by virtue of \eqref{pkphi}, there is an operator $R=\sum_{i=0}^{k'}\sum_{j\leq0}R_{ij}\p_\mu^i\p_\nu^j$ ($R_{ij}\in\cB$) such that
\[
A_{\mu,m'l'}^\nu=(M_\mu^{m'}P_\mu^{l'})_{\geq\mu-1}\la\Phi_\nu\ra\Phi_\nu^{-1} =(M_\mu^{m'}P_\mu^{l'})_{\geq\mu-1}+R\p_\mu\Phi_\nu^{-1}.
\]
Then with the same method as in \eqref{add-T}--\eqref{add-T1}, we conclude that
\begin{align}
	&\bigg(-\left[A_{\mu,m'l'}^\nu ,\, M_\nu^mP_\nu^l\right]_{\geq\nu-1}\la\Phi_\mu\ra\Phi_\mu^{-1}\nn\\
	&+\Big(\res_{\p_\nu}\left[(M_\nu^mP_\nu^l)_{\geq\nu-1},\, -(M_\mu^{m'}P_\mu^{l'})_{\geq\mu-1} \right]\Phi_\mu\p_\nu^{-1}\Phi_\mu^{-1}\Big)_{<\mu-1} \bigg)_{<\mu-1}=0. \label{addch2}
\end{align}
Thus
\[\left[\frac{\p  }{\p s_{m,l}^{(\nu)}},\,  \frac{\p  }{\p s_{m',l'}^{(\mu)}}\right]\Phi_\mu=0.\]
The other cases are similar. Therefore the theorem is proved.
\end{prf}

\begin{rmk}
The KP-mKP hierarchy is equivalent to the extended KP hierarchy \cite{GHW2024,SB2008,WZ2016} under a generic assumption. One observes that the additional symmetries of the KP-mKP hierarchy agree with those of the extended KP hierarchy (see Proposition 4.5 in \cite{LW2021} for details).
\end{rmk}

At the end of this section, we want to derive some generating series of the additional symmetries. For this purpose, let us introduce
\begin{align}
Y_{11}(\gm,\zeta)=&\sum_{m=0}^{\infty}\frac{(\zeta-\gm)^m}{m!}\sum_{l=-\infty}^{\infty}\gm^{-m-l}(M_1^mP_1^{m+l})_{<0},\label{Y11}\\
Y_{12}(\gm,\zeta)=&-\sum_{m=0}^{\infty}\frac{(\zeta-\gm)^m}{m!}\sum_{l=-\infty}^{\infty}\gm^{-m-l}(M_1^mP_1^{m+l})_{\geq0},\label{Y12}\\
Y_{21}(\gm,\zeta)=&-\sum_{m=0}^{\infty}\frac{(\zeta-\gm)^m}{m!}\sum_{l=-\infty}^{\infty}
\gm^{-m-l}(M_2^mP_2^{m+l})_{\geq1},\label{Y21} \\
Y_{22}(\gm,\zeta)=&\sum_{m=0}^{\infty}\frac{(\zeta-\gm)^m}{m!}\sum_{l=-\infty}^{\infty} \gm^{-m-l}(M_2^mP_2^{m+l})_{<1}. \label{Y22}
\end{align}
\begin{prp}\label{thm-Ywave}
The operators \eqref{Y11}--\eqref{Y22} act on Baker-Akhiezer functions of the KP-mKP hierarchy as follows:
\begin{align}
Y_{11}(\gm,\zeta)w_1(z)=&-\left(w_1(\zeta)\p_1^{-1}\cdot e^{-\beta}\p_1e^{\beta}w_1^\dag(\gm)\right)w_1(z),\label{Y11w1}\\
Y_{12}(\gm,\zeta)w_2(z)=&-w_1(\zeta)\p_1^{-1} \left(e^{-\beta}\p_1e^{\beta}w_1^\dag(\gm)\cdot w_2(z)\right),\label{Y12w2}
\\
Y_{21}(\gm,\zeta)w_1(z)=&w_2(\zeta)\p_2^{-1}\left(
w_2^{\dag}(\gm)\p_2 w_1(z) \right),\label{Y21w1}
\\
Y_{22}(\gm,\zeta)w_2(z)=&\left(w_2(\zeta)\p_2^{-1}\cdot w_2^\dag(\gm)\p_2\right)w_2(z).\label{Y22w2}
\end{align}
\end{prp}
\begin{prf}
Let us check the above equalities case by case.
\\
(i) By using Lemma \ref{thm-exres} and the second equality of \eqref{PM}, one has
\begin{align}
&(M_1^mP_1^{m+l})_{<0} \nn\\
=&\sum_{i=1}^{\infty}\p_1^{-i}\res_{\p_1}\p_1^{i-1}M_1^mP_1^{l+m}\nn\\
=&\sum_{i=1}^{\infty}\p_1^{-i}\res_{\p_1}\left(\p_1^{i-1}M_1^m\Phi_1\p_1^{l+m} \cdot\p_1^{-1}\p_1\Phi_1^{-1}e^{\beta}\p_1^{-1}e^{-\beta} e^{\beta}\p_1e^{-\beta} \right)
\nn\\
=&\sum_{i=1}^{\infty}\p_1^{-i}\res_{\ve}\left(\p_1^{i-1}M_1^m\Phi_1\p_1^{l+m}e^{\xi(\bt_1;\ve)}\cdot e^{-\beta}\p_1e^{\beta}(\p_1\Phi_1^{-1}e^{\beta}\p_1^{-1}e^{-\beta} )^*\p_1^{-1}e^{-\xi(\bt_1;\ve)}\right)
\nn\\
=&-\sum_{i=1}^{\infty}\p_1^{-i}\res_{\ve}\left(\ve^{l+m-1}\p_1^{i-1}\p_\ve^m w_1(\ve)\cdot e^{-\beta}\p_1e^{\beta}w_1^\dag(\ve)\right)
\nn\\
=&-\res_{\ve}\left(\ve^{l+m-1}\p_\ve^mw_1(\ve)\p_1^{-1}\cdot e^{-\beta}\p_1e^{\beta}w_1^\dag(\ve)\right).\nn 
\end{align}
Here $\res_{\ve}\ve^i=\delta_{i,-1}$, and the last equality is due to the fact
\begin{align*}
	f\p_\nu^{-1}=\sum_{i=1}^{\infty}\p_\nu^{-i}\cdot \p_\nu^{i-1}(f) 
\end{align*}
with $f$ being an arbitrary function of the variables $x_1$ and $x_2$.
Hence we obtain
\begin{align*}
	& Y_{11}(\gm,\zeta)w_1(z) \nn\\
=&-\left(\res_{\ve}\sum_{m=0}^{\infty}\frac{(\zeta-\gm)^m}{m!}\sum_{l=-\infty}^{\infty}
	\frac{\ve^{l+m-1}}{\gm^{m+l}}\p_\ve^mw_1(\ve)\p_1^{-1}\cdot e^{-\beta}\p_1e^{\beta}w_1^\dag(\ve)\right)w_1(z)\\
	=&-\left(\sum_{l=-\infty}^{\infty}\res_{\ve}
	\frac{\ve^{l}}{\gm^{l+1}}e^{(\zeta-\gm)\p_\ve}w_1(\ve)\p_1^{-1}\cdot e^{-\beta}\p_1e^{\beta}w_1^\dag(\ve)\right)w_1(z)\\
	=&-\left(e^{(\zeta-\gm)\p_\gm}w_1(\gm)\p_1^{-1}\cdot e^{-\beta}\p_1e^{\beta}w_1^\dag(\gm)\right)w_1(z)\\
	=&-\left(w_1(\zeta)\p_1^{-1}\cdot e^{-\beta}\p_1e^{\beta}w_1^\dag(\gm)\right)w_1(z),
\end{align*}
which confirms the equality \eqref{Y11w1}.
\\
(ii) Similarly, one has
\begin{align}
&(M_2^mP_2^{m+l})_{<1} \nn\\
=&\res_{\p_2}M_2^mP_2^{m+l}\p_2^{-1}+\sum_{i=1}^{\infty}\p_2^{-i}\res_{\p_2}\p_2^{i-1}M_2^mP_2^{l+m}\nn\\
=&\res_{\p_2}M_2^m\Phi_2\p_2^{m+l-1}\p_2\Phi_2^{-1}\p_2^{-1}\nn\\
&+\sum_{i=1}^{\infty}\p_2^{-i}\res_{\p_2}\p_2^{i-1}M_2^m\Phi_2\p_2^{l+m-1}\p_2\Phi_2^{-1}\p_2^{-1}\p_2\nn\\
=&\res_{\ve}\left(M_2^m\Phi_2\p_2^{m+l-1}e^{\xi(\bt_2;\ve)} \cdot(\p_2\Phi_2^{-1}\p_2^{-1})^*e^{-\xi(\bt_2;\ve)}\right)\nn\\
&+\sum_{i=1}^{\infty}\p_2^{-i}\res_{\ve}\left(\p_2^{i-1}M_2^m\Phi_2\p_2^{l+m-1}e^{\xi(\bt_2;\ve)} \cdot(-\p_2)(\p_2\Phi_2^{-1}\p_2^{-1})^*e^{-\xi(\bt_2;\ve)}\right) \nn\\
=&\res_{\ve}\left(\ve^{l+m-1}\p_\ve^mw_2(\ve)\p_2^{-1}\cdot\p_2\cdot w_2^\dag(\ve)\right)\nn
\\
&-\sum_{i=1}^{\infty}\p_2^{-i}\res_{\ve}\left( \ve^{l+m-1}\p_2^{i-1}\p_\ve^mw_2(\ve)\cdot \p_2 w_2^\dag(\ve)\right)\nn\\
=&\res_{\ve}\ve^{l+m-1}\p_\ve^mw_2(\ve)\p_2^{-1}\cdot\left(\p_2\cdot w_2^\dag(\ve)-\p_2(w_2^\dag(\ve))\right)\nn\\
=&\res_{\ve}\left(\ve^{l+m-1}\p_\ve^mw_2(\ve)\p_2^{-1}\cdot w_2^\dag(\ve)\p_2\right).\nn
\end{align}
Hence
\begin{align*}
	Y_{22}(\gm,\zeta)w_2(z)=&\res_{\ve}\left(\sum_{m=0}^{\infty}\frac{(\zeta-\gm)^m}{m!}\sum_{l=-\infty}^{\infty}
\frac{\ve^{l+m-1}}{\gm^{m+l}}\p_\ve^mw_2(\ve)\p_2^{-1}\cdot w_2^\dag(\ve)\p_2\right)w_2(z)
\\
=&\res_{\ve}\left(\sum_{l=-\infty}^{\infty}
\frac{\ve^{l}}{\gm^{l+1}}e^{(\zeta-\gm)\p_\ve}w_2(\ve)\p_2^{-1}\cdot w_2^\dag(\ve)\p_2\right)w_2(z)\\
=&\left(e^{(\zeta-\gm)\p_\gm}w_2(\gm)\p_2^{-1}\cdot w_2^\dag(\gm)\p_2\right)w_2(z)\\
=&\left(w_2(\zeta)\p_2^{-1}\cdot w_2^\dag(\gm)\p_2\right)w_2(z),
\end{align*}
which is just the equality \eqref{Y22w2}.
\\
(iii) Since
\begin{align}
&	(M_1^mP_1^{m+l})_{\geq0}w_2(z)\nn \\
=&\sum_{i\geq0}\left(\res_{\p_1}
	M_1^mP_1^{l+m}\p_1^{-i-1}\right) \p_1^i(w_2(z))\nn\\
	=&\sum_{i\geq0}\left(\res_{\p_1} M_1^m\Phi_1\p_1^{l+m-1} (\p_1\Phi_1^{-1}e^{\beta}\p_1^{-1}e^{-\beta}e^{\beta}\p_1e^{-\beta}\p_1^{-i-1})\right) \p_1^i(w_2(z))\nn\\ 
	=&\sum_{i\geq0}\res_{\ve}\left(
	\ve^{l+m-1}\p_\ve^mw_1(\ve)\cdot (-\p_1)^{-i-1} e^{-\beta}(-\p_1)e^{\beta}w_1^\dag(\ve) \right)\p_1^i(w_2(z))\nn\\
	=&\res_{\ve}\left(
	\ve^{l+m-1}\p_\ve^mw_1(\ve)\cdot \sum_{i\geq0}(-\p_1)^i(w_2(z)) \p_1^{-i-1} e^{-\beta}\p_1e^{\beta}w_1^\dag(\ve) \right) \nn\\
	=&\res_{\ve}\left(
	\ve^{l+m-1}\p_\ve^mw_1(\ve)\cdot \p_1^{-1} \left(w_2(z)e^{-\beta}\p_1e^{\beta}w_1^\dag(\ve)\right)\right),\nn
\end{align}
then we have
\begin{align*}
	Y_{12}(\gm,\zeta)w_2(z)=&-\res_{\ve}\sum_{m=0}^{\infty}\frac{(\zeta-\gm)^m}{m!}\sum_{l=-\infty}^{\infty}
	\frac{\ve^{l+m-1}}{\gm^{m+l}} \left(\p_\ve^mw_1(\ve)\cdot \p_1^{-1} \left(w_2(z)e^{-\beta}\p_1e^{\beta}w_1^\dag(\ve)\right) \right)\\
	=&-\res_{\ve}\sum_{l=-\infty}^{\infty}
	\frac{\ve^{l}}{\gm^{l+1}}\left(e^{(\zeta-\gm)\p_\ve}w_1(\ve)\cdot \p_1^{-1} \left(w_2(z)e^{-\beta}\p_1e^{\beta}(w_1^\dag(\ve))\right)\right)\\
	=&-e^{(\zeta-\gm)\p_\gm}w_1(\gm)\cdot \p_1^{-1} \left(w_2(z)e^{-\beta}\p_1e^{\beta}(w_1^\dag(\gm))\right)\\
	=&-w_1(\zeta)\p_1^{-1} \left(w_2(z)e^{-\beta}\p_1e^{\beta}(w_1^\dag(\gm)\right),
\end{align*}
which shows the equality \eqref{Y12w2}.
\\
(iv) Similarly, from
\begin{align}
&(M_2^mP_2^{m+l})_{\geq1} w_1(z) \nn\\
=&\sum_{i\geq1}\left(\res_{\p_2}
M_2^mP_2^{l+m}\p_2^{-i-1}\right) \p_2^i(w_1(z))\nn\\
=&\sum_{i\geq1}\left(\res_{\p_2}
M_2^m\Phi_2\p_2^{l+m-1}\p_2\Phi_2^{-1}\p_2^{-1}\p_2^{-i}\right) \p_2^i(w_1(z))\nn\\
=&\sum_{i\geq1} \res_{\ve}\left(
M_2^m\Phi_2\p_2^{l+m-1}e^{\xi(\bt_2;\ve)}\cdot (-\p_2)^{-i}(\p_2\Phi_2^{-1}\p_2^{-1})^*e^{-\xi(\bt_2;\ve)}\right) \p_2^i(w_1(z))\nn\\
=&\sum_{i\geq1}\res_{\ve}\left(
\ve^{l+m-1}\p_\ve^{m}w_2(\ve)\cdot (-\p_2)^{-i}(w_2^{\dag}(\ve))\right)
 \p_2^i(w_1(z))\nn\\
=&-\left(\res_{\ve}\left(
\ve^{l+m-1}\p_\ve^{m}w_2(\ve)\cdot \p_2^{-1}\left(\p_2(w_1(z))
w_2^{\dag}(\ve)\right)\right)\right),\nn
\end{align}
we have
\begin{align*}
Y_{21}(\gm,\zeta)w_1(z)=&\res_{\ve}\sum_{m=0}^{\infty}\frac{(\zeta-\gm)^m}{m!}\sum_{l=-\infty}^{\infty}
\frac{\ve^{l+m-1}}{\gm^{m+l}}\left(\p_\ve^mw_2(\ve)\cdot\p_2^{-1}\left(\p_2(w_1(z))
w_2^{\dag}(\ve)\right)\right)
\\
=&\res_{\ve}\sum_{l=-\infty}^{\infty}
\frac{\ve^{l}}{\gm^{l+1}}\left(e^{(\zeta-\gm)\p_\ve}w_2(\ve)\cdot\p_2^{-1}\left(\p_2(w_1(z))
w_2^{\dag}(\ve)\right)\right)
\\
=&e^{(\zeta-\gm)\p_\gm}w_2(\gm)\cdot \p_2^{-1}\left(\p_2(w_1(z))
w_2^{\dag}(\gm)\right)\\
=&w_2(\zeta)\p_2^{-1}\left(\p_2(w_1(z))
w_2^{\dag}(\gm)\right),
\end{align*}
which confirms the equality \eqref{Y21w1}. Therefore the proposition is proved.
\end{prf}


\section{The Adler-Shiota-van Moerbeke formula} \label{sec-addtau}
In this section, we want to represent the additional symmetries of the KP-mKP hierarchy into certain linear actions on the tau functions. For this purpose, we need to derive the ASvM formula.

For $\nu\in\{1,2\}$, recall $\xi(\bt_\nu; z)=\sum_{k\in\Z_{\geq1}}t_{\nu,k} z^k$ and denote
\begin{equation}\label{}
  G_\nu(z)=\exp\left(-\sum_{k=1}^{\infty}\frac{z^{-k}}{k}\frac{\p}{\p t_{\nu,k}}\right).
\end{equation}
Let us introduce the following vertex operators with parameters $\gm$ and $\zeta$:
\begin{equation}
 X_{\nu}(\gm,\zeta)=e^{-\xi(\bt_\nu; \gm)+\xi(\bt_\nu; \zeta)}G_\nu^{-1}(\gm)G_\nu(\zeta).\label{vectorde}
\end{equation}
Consider the following deformations of tau functions $\tau_\nu=\tau_\nu(\bt_1,\bt_2)$ of the KP-mKP hierarchy with infinitesimal parameters $\epsilon$ as
\begin{equation}\label{Xtau}
\tau_\nu\mapsto \tau_{\nu}+\epsilon\left(\frac{\zeta}{\gm}\right)^{1-\delta_{\mu,\nu}} X_{\mu}(\gm,\zeta)\tau_{\nu}+O(\epsilon^2).
\end{equation}
Thanks to Theorem~\ref{thm-ble2}, such deformations induce the following infinitesimal actions on the Baker-Akhiezer functions:
\begin{align}\label{XW}
X_\mu(\gm,\zeta)\star w_\nu(z)=& \bigg(\left(\frac{\zeta}{\gm}\right)^{1-\delta_{\mu\nu}} \frac{G_\nu(z)X_\mu(\gm,\zeta)\tau_\nu }{\tau_1 }  \nn\\
&-
\left(\frac{\zeta}{\gm}\right)^{1-\delta_{\mu 1}} \frac{G_\nu(z) \tau_\nu \cdot X_\mu(\gm,\zeta)\tau_1 }{\tau_1^2} \bigg) e^{\xi(\bt_\nu;z)}
\end{align}
with $\nu,\mu\in\{1,2\}$.

\begin{thm}[ASvM formula]\label{thm-asvm}
For the KP-mKP hierarchy, the following equalities hold true
\begin{align}
X_{\nu}(\gm,\zeta)\star w_{\mu}(z)=\frac{\gm-\zeta}{\gm}Y_{\nu\mu}(\gm,\zeta)w_\mu(z),\label{asvmf}
\end{align}
where $\nu,\mu \in\{1,2\}$, and the right hand side is given in Proposition~\ref{thm-Ywave}.
\end{thm}
	
\begin{prf}
Let us do the verification case by case, based on Lemmas~\ref{thm-taueq} and \ref{thm-pwave1}, as well as Proposition~\ref{thm-Ywave}.
\\
(i) It is straightforward to show
\begin{align}
	&\frac{\gm}{\gm-\zeta}X_1(\gm,\zeta)\star w_1(z) \nn\\
	=& \frac{\gm}{\gm-\zeta} e^{-\xi(\bt_1;\gm)+\xi(\bt_1;\zeta)+\xi(\bt_1;z)} \left( \frac{1-\zeta/z}{1-\gm/z} \cdot
\frac{\tau_1(\bt_1+[\gm^{-1}]-[\zeta^{-1}]-[z^{-1}],\bt_2) }{\tau_1(\bt_1,\bt_2)} \right. \nn\\
	&\left.
	 -\frac{\tau_1(\bt_1-[z^{-1}],\bt_2)\tau_1(\bt_1+[\gm^{-1}]-[\zeta^{-1}],\bt_2)}{\tau_1(\bt_1,\bt_2)^2} \right)
\nn\\
=& \frac{\gm}{\gm-\zeta} e^{-\xi(\bt_1;\gm)+\xi(\bt_1;\zeta)+\xi(\bt_1;z)} \bigg( \frac{z-\zeta}{z-\gm} \cdot
\nn\\
&\cdot
 \left( \frac{1/\zeta(1/z-1/\gm)}{1/\gm(1/z-1/\zeta)} \frac{\tau_1(\bt_1+[\gm^{-1}]-[\zeta^{-1}],\bt_2)\tau_1(\bt_1-[z^{-1}],\bt_2)}{\tau_1(\bt_1,\bt_2)^2} \right.
 \nn\\
 &\left.+\frac{1/z(1/\zeta-1/\gm)}{1/\gm(1/\zeta-1/z)}   \frac{\tau_1(\bt_1+[\gm^{-1}]-[z^{-1}],\bt_2)\tau_1(\bt_1-[\zeta^{-1}],\bt_2)}{\tau_1(\bt_1,\bt_2)^2} \right)
  \nn\\
	&
	 -\frac{\tau_1(\bt_1+[\gm^{-1}]-[\zeta^{-1}],\bt_2)\tau_1(\bt_1-[z^{-1}],\bt_2)}{\tau_1(\bt_1,\bt_2)^2} \bigg)
\nn\\
=& \frac{\gm}{z-\gm} e^{-\xi(\bt_1;\gm)+\xi(\bt_1;\zeta)+\xi(\bt_1;z)}
\frac{\tau_1(\bt_1+[\gm^{-1}]-[z^{-1}],\bt_2)\tau_1(\bt_1-[\zeta^{-1}],\bt_2)}{\tau_1(\bt_1,\bt_2)^2}
\nn\\
	=&-\Big(w_1(\zeta)\p_1^{-1}\cdot  e^{-\beta}\p_1 e^{\beta}w_1^\dag(\gm)\Big)w_1(z) =Y_{11}(\gm,\zeta)w_1(z).\nn
\end{align}	
Here the second equality is due to \eqref{kp-fay1}, the third equality is due to \eqref{pwave1}, and last equality due to \eqref{Y11w1}.
			
(ii) Similarly, by using \eqref{tau21}, \eqref{pwave2} and \eqref{Y12w2}, one has
\begin{align}
	&\frac{\gm}{\gm-\zeta}X_1(\gm,\zeta)\star w_2(z)
\nn\\
=&\frac{\gm}{\gm-\zeta}e^{-\xi(\bt_1;\gm)+\xi(\bt_1;\zeta)+\xi(\bt_2;z)}\bigg(\frac{\zeta}{\gm}\cdot \frac{\tau_2(\bt_1+[\gm^{-1}]-[\zeta^{-1}],\bt_2-[z^{-1}])\tau_1(\bt_1,\bt_2)}{\tau_1(\bt_1,\bt_2)^2} \nn\\
	& - \frac{\tau_2(\bt_1,\bt_2-[z^{-1}])\tau_1(\bt_1+[\gm^{-1}]-[\zeta^{-1}],\bt_2)}{\tau_1(\bt_1,\bt_2)^2} \bigg)
\nn\\
=&\frac{\gm}{\gm-\zeta}e^{-\xi(\bt_1;\gm)+\xi(\bt_1;\zeta)+\xi(\bt_2;z)} \frac{1}{\tau_1(\bt_1,\bt_2)^2}\bigg(\frac{\zeta}{\gm}
\left(  \frac{\zeta-\gm}{\zeta} \tau_1(\bt_1-[\zeta^{-1}],\bt_2)\tau_2(\bt_1+[\gm^{-1}],\bt_2-[z^{-1}]) \right. \nn\\
& \left. + \frac{\gm}{\zeta}\tau_1(\bt_1+[\gm^{-1}]-[\zeta^{-1}],\bt_2)\tau_2(\bt_1,\bt_2-[z^{-1}]) \right)
\nn\\
&-\tau_1(\bt_1+[\gm^{-1}]-[\zeta^{-1}],\bt_2)\tau_2(\bt_1,\bt_2-[z^{-1}]) \bigg)
\nn\\
=&- \frac{\tau_1(\bt_1-[\zeta^{-1}],\bt_2)}{\tau_1(\bt_1,\bt_2)}e^{\xi(\bt_1;\zeta)-\xi(\bt_1;\gm)+\xi(\bt_2;z)} \frac{\tau_2(\bt_1+[\gm^{-1}],\bt_2-[z^{-1}])}{\tau_1(\bt_1,\bt_2)}\nn\\
	=&-w_1(\zeta)\p_1^{-1}\Big(e^{-\beta}\p_1e^{\beta} w_1^\dag(\gm)\cdot  w_2(z)\Big)=Y_{12}(\gm,\zeta)w_2(z).\nn
\end{align}

The other two cases $(\nu,\mu)=(2,1)$ and  $(\nu,\mu)=(2,2)$  are similar, with the help of  \eqref{tau12}, \eqref{pwave4}, \eqref{Y21w1} and of \eqref{tau121}, \eqref{pwave3}, \eqref{Y22w2} respectively.
Therefore, the theorem is proved.
\end{prf}
	
Now we are ready to represent the additional symmetries of the KP-mKP hierarchy via its tau functions. To this end, let us expand the vertex operators $X_\nu(\gm,\zeta)$ into Laurent series:
\begin{align}
X_\nu(\gm,\zeta)=\sum_{m=0}^{\infty}\frac{(\zeta-\gm)^m}{m!}\sum_{l=-\infty}^{\infty}\gm^{-m-l}W_{m,l}^{(\nu)}, \label{vectorde1}
\end{align}
where
\begin{align}\label{W}
W_{m,l}^{(\nu)}=\res_{\gm=0}\left(\gm^{m+l-1}\left.\frac{\p^m}{\p \zeta^m}\right|_{\zeta=\gm}X_\nu(\gm,\zeta) \right).
\end{align}
For example, one has
\begin{align} \label{W012}
W_{0,l}^{(\nu)}=\delta_{l,0},\quad W_{1,l}^{(\nu)}=p_{\nu,l},\quad W_{2,l}^{(\nu)}=\sum_{i+j=l}:p_{\nu,i}p_{\nu,j}:-(l+1)p_{\nu,l},
\end{align}
where
\begin{align}
p_{\nu,k}=\left\{\begin{array}{cl}
\dfrac{\p}{\p t_{\nu,k}}, & k\in\Z_{>0}; \\
	|k|t_{\nu,|k|},& k\in\Z_{<0}; \\
0, & k=0,
\end{array}
\right. \label{pnuk}
\end{align}
and ``$:~ :$'' stands for the normal order product, namely, operators $p_{\nu,k>0}$ need to be put to the right of  $p_{\nu,k<0}$.
\begin{thm}\label{thm-tau}
The additional symmetries of the KP-mKP hierarchy satisfy
\begin{align}
	\frac{\p \tau_\nu }{\p s_{m,m+l}^{(\nu)}}=&\left(\frac{W_{m+1,l}^{(\nu)}}{m+1}+\delta_{l0}c^{(\nu)}_m\right)\tau_\nu ,\label{addtau1}\\
	\frac{\p \tau_\mu }{\p s_{m,m+l}^{(\nu)}}=&\left(\frac{W_{m+1,l}^{(\nu)}}{m+1}+W_{m,l}^{(\nu)}+\delta_{l0}c^{(\nu)}_m\right)\tau_\mu ,\label{addtau2}
\end{align}
where $\nu,\mu\in\{1,2\}$, $\nu\neq\mu$,  $m\in\Z_{\geq0}$, $l\in\Z$ and $c^{(\nu)}_m$ are constants.
\end{thm}
\begin{prf}
Thanks to the expressions \eqref{vectorde1} and \eqref{w2tau}, we have
\begin{align}
X_2(\gm,\zeta)\star w_2(z)
=&\left(\frac{G_2(z) X_2(\gm,\zeta)\tau_2 }{\tau_1 }-\frac{\zeta}{\gm}\frac{G_2(z)\tau_2 \cdot X_2(\gm,\zeta)\tau_1 }{\tau_1 ^2} \right)e^{\xi(\bt_2;z)}
\nn\\
=&w_2(z)\left(G_2(z)\frac{X_2(\gm,\zeta)\tau_{2}}{\tau_{2}}-\left(1+\frac{\zeta-\gm}{\gm}\right)\frac{X_2(\gm,\zeta)\tau_{1}}{\tau_{1}} \right)
\nn\\
=&w_2(z)\sum_{m=-1}^{\infty}\frac{(\zeta-\gm)^{m+1}}{(m+1)!}\sum_{l=-\infty}^{\infty}\gm^{-m-1-l} \left(G_2(z)\frac{W_{m+1,l}^{(2)}\tau_2}{\tau_2}-\frac{W_{m+1,l}^{(2)}\tau_1}{\tau_1}\right)
\nn\\
&-w_2(z)\sum_{m=0}^{\infty}\frac{(\zeta-\gm)^{m+1}}{m!}\sum_{l=-\infty}^{\infty}\gm^{-m-1-l} \frac{W_{m,l}^{(2)}\tau_1}{\tau_1}
\nn\\
=& w_2(z)\sum_{m=0}^{\infty}\frac{(\zeta-\gm)^{m+1}}{m!}\sum_{l=-\infty}^{\infty}\gm^{-m-1-l} \cdot
\nn\\
&\qquad \cdot\left(G_2(z)\frac{W_{m+1,l}^{(2)}\tau_2}{(m+1)\tau_2} -\frac{W_{m+1,l}^{(2)}\tau_1}{(m+1)\tau_1}-\frac{W_{m,l}^{(2)}\tau_1}{\tau_1}\right), \label{X2w22}
\end{align}
where the last equality is due to $W_{0,l}^{(\nu)}=\delta_{l,0}$.
On the other hand, by using \eqref{add1} and \eqref{Y11}, we have
\begin{align}
\frac{\gm-\zeta}{\gm}Y_{22}(\gm,\zeta)w_2(z) &=-\frac{\gm-\zeta}{\gm}\sum_{m=0}^{\infty}\frac{(\zeta-\gm)^m}{m!}\sum_{l=-\infty}^{\infty}\gm^{-m-l}\frac{\p w_2(z)}{\p s_{m,m+l}^{(2)}}
\nn\\
&=\sum_{m=0}^{\infty}\frac{(\zeta-\gm)^{m+1}}{m!}\sum_{l=-\infty}^{\infty}\gm^{-m-1-l}w_2(z)\left(G_2(z)\frac{\p \log\tau_2}{\p s_{m,m+l}^{(2)}}-\frac{\p\log \tau_1}{\p s_{m,m+l}^{(2)}}  \right).\label{Y22w22}
\end{align}
According to the ASvM formula \eqref{asvmf} for $\nu=\mu=2$, by comparing coefficients in the Laurent series \eqref{X2w22} and \eqref{Y22w22} we obtain
\begin{align*}
 \frac{\p \log\tau_2}{\p s_{m,m+l}^{(2)}}-\frac{W_{m+1,l}^{(2)}\tau_2}{(m+1)\tau_2} = \frac{\p\log \tau_1}{\p s_{m,m+l}^{(2)}}-\frac{W_{m+1,l}^{(2)}\tau_1}{(m+1)\tau_1}-\frac{W_{m,l}^{(2)}\tau_1}{\tau_1} =c^{(2)}_{m,m+l},
\end{align*}
where  $c^{(2)}_{m,m+l}$ are certain  constants.
Similarly,  the ASvM formula \eqref{asvmf} with $\nu=2$ and $\mu=1$ leads to
\begin{align*}
 \frac{\p \log\tau_2}{\p s_{m,m+l}^{(1)}}-\frac{W_{m+1,l}^{(1)}\tau_2}{(m+1)\tau_2}-\frac{W_{m,l}^{(1)}\tau_2}{\tau_2} = \frac{\p\log \tau_1}{\p s_{m,m+l}^{(1)}}-\frac{W_{m+1,l}^{(1)}\tau_1}{(m+1)\tau_1} =c^{(1)}_{m,m+l}
\end{align*}
with constants $c^{(1)}_{m,m+l}$. Without loss of generality, the tau functions $\tau_\nu$ can be replaced by $\tau_\nu \exp\left(-\sum_{\mu=1}^2\sum_{m\ge0; \, l\ne 0} c^{(\mu)}_{m,m+l}s^{(\mu)}_{m,m+l}\right)$, thus we confirm the equalities \eqref{addtau1} and \eqref{addtau2}.
The theorem is proved.
\end{prf}
	
\begin{rmk}
It is known that the KP-mKP hierarchy can be reduced to the KP and the mKP hierarchies \cite{GHW2024}. Under such reductions, the formulae \eqref{addtau1}--\eqref{addtau2} coincide with the results for the KP hierarchy \cite{ASvM94} and that for the mKP hierarchy \cite{Cheng2018}.
\end{rmk}	

\begin{rmk}
As was shown by Adler and van Moerbeke, certain vertex operators $\mathbb{X}_{1 2}$ and $\mathbb{X}_{2 1}$ that
contain two different series of time variables acting on the tau functions give symmetries of the 
two-Toda lattice hierarchy (see Theorem 7.1 in \cite{AvM99}). Contrastly, on the left hand side of the ASvM formula \eqref{asvmf}, the vertex operator $X_\nu(\gamma,\zeta)$ with $\nu=1$ or $2$ involves only one series of time variables $\mathbf{t}_\nu$ (recall \eqref{vectorde}), which is sufficient to represent the additional symmetries (4.5) of the KP-mKP hierarchy via its tau functions. Inspired by \cite{AvM99}, it is interesting to study whether the ASvM formula of the KP-mKP hierarchy can be extended to a version that also involves some vertex operators containing both $\mathbf{t}_1$ and $\mathbf{t}_2$.
\end{rmk}

\section{Virasoro constraints to the Burgers-KdV hierarchy and its higher order extensions}
Now let us apply the above results to study the Virasoro constraints to certain subhierarchies of the KP-mKP hierarchy that were considered in \cite{GHW2025}.

Given a positive integer $r\ge2$, in the present section we assume the KP-mKP hierarchy to satisfy the following conditions:
\begin{align}
(\Phi_1\p_1^r\Phi_1^{-1})_{<0}=0,\quad \Phi_2=e^\beta, \quad \p_2(\beta)=0. \label{crKP}
\end{align}
Then the KP-mKP hierarchy is reduced to
\begin{align}
\frac{\p\Phi_1}{\p t_{1,k}}=-(\Phi_1\p_1^k\Phi_1^{-1})_{<0}\Phi_1, \quad \frac{\p e^\beta}{\p t_{1,k}}=(\Phi_1\p_1^k\Phi_1^{-1})_{\ge0}(e^\beta), \quad  \frac{\p\Phi_1}{\p t_{2,k}}=0, \quad \frac{\p e^\beta}{\p t_{2,k}}=0
\end{align}
with $k\in\Z_{\ge1}$. Or equivalently, denote $L=\Phi_1\p_1^r\Phi_1^{-1}$ and $t_k=t_{1,k}$, then one obtains the $(r,0)$-reduced hierarchy \cite{GHW2025}:
\begin{align}
\frac{\p L}{\p t_{k}}=[(L^{\frac{k}{r}})_{\geq0},\,L],\quad \frac{\p e^{\beta}}{\p t_{k}}=(L^{\frac{k}{r}})_{\geq0}(e^{\beta}),\qquad k\in\Z_{\ge1}. \label{drKP}
\end{align}
This hierarchy with $r=2$ is called the Burgers-KdV hierarchy \cite{Bur2015}, which is an extension of the open KdV hierarchy that governs the open intersection numbers \cite{Bur2016,PST}. While $r\ge3$, such a hierarchy is called the extended $r$-reduced KP hierarchy \cite{BY2015}.

For the reduced hierarchy \eqref{drKP} there exit two tau functions $\tau_1$ and $\tau_2$ that are inherited from those of the KP-mKP hierarchy. In particular, it is easy to see that
\[
\frac{\p\Phi_1}{\p t_{k}}=0, \quad \frac{\p L}{\p t_{k}}=0, \quad \frac{\p\tau_1}{\p t_{k}}=0
\]
whenever $k$ is a multiple of $r$ (written as $r|k$). Moreover, one has
\begin{align}
\frac{\p e^\beta}{\p t_{r l}}=L^{l}(e^\beta)=\frac{\p^l}{\p t_{r}^l}(e^\beta), \quad l\in\Z_{\ge1}.
\end{align}

Let us consider the case $r=2$ first. Noting $P_1=\Phi_1\p_1\Phi_1^{-1}=L^{1/2}$ and recalling $M_1$ in \eqref{Mmu}, we let
\begin{align}
	S_j=\frac{1}{2}M_1P_1^{2j+1}+\frac{4j+3}{4}P_1^{2j},\label{2KPS}
\end{align}
and introduce the following evolutionary equations:
\begin{align}
\frac{\p \Phi_1}{\p \sg_{j}}=-(S_j)_{<0}\Phi_1,\quad \frac{\p e^{\beta}}{\p \sg_{j}}=(S_j)_{\geq0}(e^\beta),\qquad j\in\Z_{\ge-1}.\label{eKPadd}
\end{align}
It is easy to see that these equations are consistent with the conditions \eqref{crKP}, and that they are indeed reductions of linear combinations of the additional symmetries \eqref{addde} of the KP-mKP hierarchy, say,
\begin{align}
\frac{\p}{\p \sg_{j}}=\frac{1}{2}\frac{\p}{\p s_{1,1+2j}^{(1)}}+\frac{4j+3}{4}\frac{\p}{\p s_{0,2j}^{(1)}}. \label{BKdVsym}
\end{align}
Hence the equations \eqref{eKPadd} define a series of symmetries of the Burgers-KdV hierarchy, which are called the \emph{Virasoro symmetries} due to the following proposition.
\begin{prp}\label{thm-VirBKdV}
For the Burgers-KdV hierarchy, the symmetries \eqref{BKdVsym} satisfy
\begin{align}
\frac{\p \tau_\nu}{\p \sg_{j}}=V_{\nu,j}\tau_\nu, \quad j\in\Z_{\ge-1},\, \nu\in\{1,2\}.\label{Vtau}
\end{align}
Here the operators $V_{\nu,j}$ read (recall $p_{1,k}$ given in \eqref{pnuk})
\begin{align}
V_{1,j}=&\frac{1}{4}\sum_{a\in\Z^{\mathrm{odd}}}:p_{1,a}p_{1,2j-a}:+\delta_{j,0}\frac{1}{16}, \label{Vj1}
\\
V_{2,j}=&\frac{1}{4}\sum_{a\in\Z}:p_{1,a}p_{1,2j-a}:+\frac{j+2}{2}p_{1,2j}+\delta_{j,0}\frac{13}{16},\label{Vj2}
\end{align}
and they obey the Virasoro commutation relations
\begin{align}
[V_{\nu,j},\, V_{\nu,k} ]=(j-k)V_{\nu,j+k}, \quad j,k\in\Z_{\ge-1},\, \nu\in\{1,2\}.\label{Vjk}
\end{align}
\end{prp}
\begin{prf}
Firstly, according to equalities \eqref{BKdVsym}, \eqref{W012} and Theorem \ref{thm-tau}, one has
\begin{align*}
\frac{\p \tau_1}{\p \sg_{j}}
=&\left(\frac{1}{4}W_{2,2j}^{(1)}+\frac{1}{2}\delta_{j,0}c_1^{(1)}+\frac{4j+3}{4}W_{1,2j}^{(1)}+\frac{4j+3}{4} \delta_{j,0}c_0^{(1)}\right)\tau_1\nn\\
=&\left(\frac{1}{4}\sum_{a\in\Z}:p_{1,a}p_{1,2j-a}:+\frac{j+1}{2}p_{1,2j}+ \delta_{j,0}c\right)\tau_1.
\end{align*}
Note that the linear term $\frac{j+1}{2}p_{1,2j}$ vanishes when $j\in\{-1,0\}$, and that $\p\tau_1/\p t_{2 j}=0$ for $j\ge1$. Hence by taking $c={1}/{16}$ we obtain \eqref{Vtau} with $\nu=1$.

In the same way, we have
\begin{align}
\frac{\p \tau_2}{\p \sg_{j}}
=&\frac{1}{2}\left(\frac{1}{2}W_{2,2j}^{(1)}+W_{1,2j}^{(1)}+ \delta_{j,0}c_1^{(1)}\right)\tau_2 +\frac{4j+3}{4}\left(W_{1,2j}^{(1)}+W_{0,2j}^{(1)}+ \delta_{j,0}c_0^{(1)} \right)\tau_2\nn\\
=&\left(\frac{1}{4}\sum_{a\in\Z}:p_{1,a}p_{1,2j-a}:+\frac{j+2}{2}p_{1,2j}+ \delta_{j,0}\frac{13}{16}\right)\tau_2,
\end{align}
which confirms \eqref{Vtau} with $\nu=2$.

It is straightforward to verify the commutation relations \eqref{Vjk}. Therefore the proposition is proved.
\end{prf}


According to the Proposition, one has
\[
\frac{\p \tau_1}{\p \sg_{-1}}=\sum_{a\in\Z_{\geq1}}\frac{2a+1}{2}t_{2a+1}\frac{\p\tau_1}{\p t_{2a-1}}+\frac{t_{1}^2}{4}.
\]
Moreover, from $\tau_2=e^{\beta}\tau_1$ and \eqref{Vtau} it follows that
\begin{align}
\frac{\p e^\beta}{\p \sg_{-1}}=&\frac{1}{\tau_1}\frac{\p \tau_2}{\p \sg_{-1}}-\frac{\tau_2}{\tau_1^2}\frac{\p \tau_1}{\p \sg_{-1}}\nn\\
=&\frac{1}{\tau_1}\left(\sum_{a\in\Z_{\geq1}}\frac{2a+1}{2}t_{2a+1}\frac{\p}{\p t_{2a-1}}+\sum_{a\in\Z_{\geq2}}a t_{2a}\frac{\p}{\p t_{2a-2}}+ \frac{t_{1}^2}{4}+t_{2}\right)\tau_2\nn\\
&-\frac{\tau_2}{\tau_1^2}\left(\sum_{a\in\Z_{\geq1}}\frac{2a+1}{2}t_{2a+1}\frac{\p}{\p t_{2a-1}}+\frac{t_{1}^2}{4} \right)\tau_1\nn\\
=&\left(\sum_{a\in\Z_{\geq3}}\frac{a}{2}t_{a}\frac{\p }{\p t_{a-2}}+t_{2}\right)e^{\beta},\nn
\end{align}
namely,
\begin{align}
\frac{\p \beta}{\p \sg_{-1}}=\sum_{a\in\Z_{\geq1}}\frac{a+2}{2}t_{a+2}\frac{\p \beta}{\p t_{a}}+t_{2}.\nn
\end{align}
\begin{thm}\label{thm-2KPvir}
Given a constant $\kappa$, suppose that a solution of the Burgers-KdV hierarchy satisfies the following string equations:
\begin{align}
\frac{\p \tau_1}{\p \sg_{-1}}=\kappa\frac{\p\tau_1}{\p t_{1}},\quad \frac{\p \beta}{\p \sg_{-1}}=\ka\frac{\p\beta}{\p t_{1}}. \label{2KPstring}
\end{align}
Then the following Virasoro constraints hold:
\begin{align}
\left(V_{\nu,j}-\ka\frac{\p}{\p t_{2j+3}}\right)\tau_\nu=0,\quad j\in\Z_{\ge-1},\,\nu\in\{1,2\}.\label{2KPvir}
\end{align}
\end{thm}
\begin{prf}
Since $\left(\frac{\p }{\p \sg_{-1}}-\ka\frac{\p}{\p t_{1}}\right)\tau_1=0$, then
\[
\left(-S_{-1}+\ka P_{1}\right)_{<0}\Phi_1=\left(\frac{\p }{\p \sg_{-1}}-\ka \frac{\p}{\p t_{1}}\right)\Phi_1=0,
\]
which leads to
\begin{equation}\label{SPm1}
 (S_{-1})_{<0}=\ka(P_1)_{<0}.
\end{equation}
Hence, for $j\in\Z_{\geq0}$, by using \eqref{2KPS} and $(P_1^2)_{<0}=L_{<0}=0$ one has
\begin{align}
(S_{j})_{<0}=&\left(\frac{1}{2}M_1P_1^{2j+1}+\frac{4j+3}{4}P_1^{2j} \right)_{<0}\nn\\
=&\left(S_{-1}P_1^{2j+2}\right)_{<0}
=\left((S_{-1})_{<0}P_1^{2j+2} \right)_{<0}\nn\\
=&\left(\ka(P_1)_{<0}P_1^{2j+2} \right)_{<0}=\ka(P_1^{2j+3})_{<0}, \nn
\end{align}
which implies that
\[\left(\frac{\p }{\p \sg_{j}}-\ka\frac{\p}{\p t_{2j+3}}\right)\Phi_1=0, \]
i.e., the equalities \eqref{2KPvir} hold for $\nu=1$.

On the other hand, the second equation in \eqref{2KPstring} is equivalent to
\begin{align}
(S_{-1}-\ka P_1)_{\geq0}(e^{\beta})=0.\nn
\end{align}
Hence, for $j\ge-1$, by using \eqref{2KPS} and \eqref{eKPadd} one has
\begin{align}
\left(\frac{\p}{\p \sg_{j}}-\ka\frac{\p }{\p t_{2j+3}} \right)e^{\beta}=&\left(S_j-\ka P_1^{2j+3} \right)_{\ge0}(e^{\beta})\nn\\
=&\left((S_{-1}-\ka P_1)P_1^{2j+2}+(j+1)P_1^{2j} \right)_{\geq0}(e^{\beta})\nn\\
=&\left(P_1^{2j+2}(S_{-1}-\ka P_1)+[S_{-1},\,P_1^{2j+2} ]+(j+1)P_1^{2j} \right)_{\geq0}(e^{\beta})\nn\\
=&P_1^{2j+2}(S_{-1}-\ka P_1)_{\geq0}(e^{\beta})=0,\nn
\end{align}
where the third equality is due to
\begin{align}
[S_{-1},\,P_1^{2j+2} ]=\frac{1}{2}[M_1P_1^{-1},\, P_1^{2j+2} ]=-(j+1)P_1^{2j}\nn
\end{align}
and \eqref{SPm1}.
Thus we arrive at
\begin{align}
\left(\frac{\p }{\p \sg_{j}}-\ka\frac{\p}{\p t_{2j+3}}\right)\tau_2=0, \quad j\ge-1.\nn	
\end{align}
Therefore, the theorem is proved with the help of Proposition~\ref{thm-VirBKdV}.
\end{prf}

\begin{rmk}
Since $\tau_2=\tau_1 e^\beta$, then the string equations \eqref{2KPstring} can also be represented as
\[
\left(V_{\nu,-1}-\ka\frac{\p}{\p t_{1}}\right)\tau_\nu=0,\quad \nu\in\{1,2\}.
\]
One observes that the constraints \eqref{2KPvir} with $\nu=1$ are just the Virasoro constraints to the KdV hierarchy (see for instance \cite{AvM}).
\end{rmk}

\begin{rmk}
If one writes $\tau_1=e^{F^c}$ and $\beta=F^o$ such that $\tau_2=e^{F^c+F^o}$, then the Virasoro constraints \eqref{2KPvir} with $\nu=2$ coincide with equations (1.11) in \cite{Bur2015} under the following replacements:
\begin{align}
t_{2i+1}\mapsto \frac{2^i \sqrt{2}}{(2i+1)!! u}t_i, \quad
t_{2(i+1)}\mapsto \frac{\dt_{i,0}}{u}s ,\qquad i\in\Z_{\ge0}
\end{align}
and  $\ka\mapsto \sqrt{2}/u$, where $u$ is a nonzero parameter. Such constraints were conjectured in \cite{PST} and proved by Buryak (see Theorem~1.2 in \cite{Bur2015}) with the help of certain recursion relations of the Burgers-KdV hierarchy. Here we have derived the Virasoro constraints to the Burgers-KdV hierarchy in a different way, based on the additional symmetries of the KP-mKP hierarchy.
\end{rmk}

Similarly, for $r\ge3$, we consider the following linear combinations of additional symmetries of the KP-mKP hierarchy:
\begin{align}\label{dsgr}
\frac{\p}{\p \sigma_{j}}=\frac{1}{r}\frac{\p }{\p s_{1,rj+1}^{(1)}}+\frac{2rj+r+1}{2r}\frac{\p }{\p s_{0,rj}^{(1)}},\quad j\in\Z_{\ge-1},
\end{align}
which are reduced to symmetries of the hierarchy \eqref{drKP}. With the same method as above, we obtain the following results.
\begin{thm}\label{thm-rKPvir}
 For the hierarchy \eqref{drKP}, the following statements hold true.
\begin{itemize}
\item[(i)] The symmetries \eqref{dsgr} can be represented via tau functions as
\begin{align}
	\frac{\p \tau_\nu}{\p \sigma_{j}}=V_{\nu,j}\tau_\nu,
\end{align}
where the operators
\begin{align}
V_{1,j}=&\frac{1}{2r}\sum_{a\in\Z\setminus r\Z}: p_{1,a}p_{1,r j-a}: + \delta_{j,0}\frac{r^2-1}{24r},
\\	V_{2,j}=&\frac{1}{2r}\sum_{a\in\Z}:p_{1,a}p_{1,r j-a}: +\frac{rj+r+2}{2r}p_{1,rj}+\delta_{j,0}\frac{r^2+12 r +11}{24r}
\end{align}
obey the Virasoro commutation relations
\begin{align}
	[V_{\nu,j},\, V_{\nu,k} ]=(j-k)V_{\nu,j+k}, \quad j,k\in\Z_{\ge-1},\, \nu\in\{1,2\}.\nn
\end{align}
\item[(ii)] Given a constant $\ka$, if the two tau functions solve the string equations
\begin{equation} \label{stringr}
	\left(V_{\nu,-1}-\ka\frac{\p}{\p t_{1}}\right)\tau_\nu=0,\quad \nu\in\{1,2\},
\end{equation}
then they satisfy the following Virasoro constraints:
\begin{align}
	\left(V_{\nu,j}-\ka\frac{\p}{\p t_{r j+r+1}}\right)\tau_\nu=0,\quad j\in\Z_{\ge-1},\,\nu\in\{1,2\}. \label{rKPvir}
\end{align}
\end{itemize}
\end{thm}

\begin{rmk}
Up to some constant factors for the time variables, the string equations \eqref{stringr} were solved in \cite{BY2015} to obtain a certain partition function $\tau_2=Z_E$. According to the above theorem, one sees that such a partition function also admits a series of Virasoro constraints.
\end{rmk}

\section{Concluding remarks}

For the KP-mKP hierarchy, we have constructed its additional symmetries, and derived the ASvM formula with the help of a class of (differential) Fay identities. Consequently, these additional symmetries are represented as certain linear actions on the two tau functions of the hierarchy. These results provide a unified viewpoint to understand the additional symmetries of the KP and of the modified KP hierarchies in the literature.

In our approach, we have also found some rather nontrivial properties of the KP-mKP hierarchy. First, each of the two tau functions $\tau_1$ and $\tau_2$ of this hierarchy is shown to obey the Hirota bilinear equations of the KP hierarchy with respect to either of the two series of time variables in the KP-mKP hierarchy. Second, a nontrivial factor needs to be taken into account in the actions \eqref{XW} of vertex operators on the two tau functions. We hope for a better understanding for such properties.

As an application of the additional symmetries of the KP-mKP hierarchy, we have obtained a proof of the Virasoro constraints to the Burgers-KdV hierarchy, in an approach different from that in \cite{Bur2015}. As is shown in \cite{GHW2025}, the KP-mKP hierarchy also admits certain $(n,m)$-reductions with positive integers $n$ and $m$. It is interesting to study Virasoro symmetries/constraints to such reduced hierarchies. For instance, research results of this kind might be helpful to understand the ``ghost'' symmetries of the constrained KP hierarchy \cite{ANP1997} (cf.\cite{LW2021}). This will be considered elsewhere.

{\bf Acknowledgments.}
{\noindent \small
The authors thank the anonymous referees for their very helpful comments and encouragement.
This work is partially supported by National Key R\&D Program of China 2023YFA1009801,
NSFC No.\,12471243 and Guangzhou S\&T Program No. SL2023A04J01542. }


\end{document}